\documentclass[journal]{IEEEtran}
\usepackage{amsmath,amssymb,mathtools}
\usepackage{amsthm}
\usepackage[percent]{overpic}
\theoremstyle{definition}
\newtheorem{theorem}{Theorem}
\newtheorem{remark}{Remark}
\usepackage{titlesec}
\usepackage{xcolor}
\usepackage{mathrsfs}
\definecolor{lightred}{RGB}{180,60,60}

\usepackage{graphicx}
\usepackage[caption=false,font=footnotesize]{subfig}

\usepackage[section]{placeins}

\usepackage[ruled,vlined]{algorithm2e}
\usepackage{algpseudocode}

\usepackage{comment}
\usepackage{url}
\usepackage{cite}
\usepackage[acronym,shortcuts]{glossaries}

\usepackage{soul}
\usepackage[most]{tcolorbox}

\sethlcolor{red!18}

\newtcolorbox{revblock}{
    colback=red!8,
    colframe=red!45!black,
    boxrule=0.4pt,
    arc=1.5mm,
    left=2mm,
    right=2mm,
    top=1mm,
    bottom=1mm,
    before skip=0pt,
    after skip=0pt,
    parbox=false
}

\theoremstyle{plain}

\algrenewcommand\algorithmicrequire{\textbf{Input:}}
\algrenewcommand\algorithmicensure{\textbf{Output:}}
\newtheorem{lemma}[theorem]{Lemma}

\theoremstyle{definition}

\theoremstyle{remark}
\newacronym{LDPC}{LDPC}{low-density parity check}
\newacronym{ISFFT}{ISFFT}{inverse symplectic finite Fourier transform}
\newacronym{SFFT}{SFFT}{symplectic finite Fourier transform}
\newacronym{DFT}{DFT}{discrete Fourier transform}
\newacronym{DoF}{DoF}{degrees of freedom}
\newacronym{2D}{2D}{two-dimensional}
\newacronym{1D}{1D}{one-dimensional}
\newacronym{3GPP}{3GPP}{3rd generation partnership project}
\newacronym{5G}{5G}{fifth generation}
\newacronym{5GNR}{5GNR}{fifth generation new radio}
\newacronym{6G}{6G}{sixth generation}
\newacronym{ADC}{ADC}{analog-to-digital converter}
\newacronym{ADR}{ADR}{antenna decentralized rate}
\newacronym{AER}{AER}{activity error rate}
\newacronym{AMP}{AMP}{approximate message passing}
\newacronym{ANN}{ANN}{approximate nearest neighbors}
\newacronym{AP}{AP}{access point}
\newacronym{ASB}{ASB}{adaptively scaled belief}
\newacronym{AUD}{AUD}{active user detection}
\newacronym{AWGN}{AWGN}{additive white Gaussian noise}
\newacronym{AoA}{AoA}{angle of arrival}
\newacronym{AoD}{AoD}{angle of departure}
\newacronym{BAd-VAMP}{BAd-VAMP}{bilinear adaptive VAMP}
\newacronym{BBI}{BBI}{Bayesian bilinear inference}
\newacronym{BER}{BER}{bit error rate}
\newacronym{BG}{BG}{Bernoulli-Gaussian}
\newacronym{BMMSE}{BMMSE}{Bussgang minimum mean square error}
\newacronym{BiGAMP}{BiGAMP}{bilinear generalized approximate message passing}
\newacronym{BIP}{BIP}{bilinear inference problem}
\newacronym{GAMP}{GAMP}{generalized approximate message passing}
\newacronym{GF}{GF}{grant-free}
\newacronym{BiGaBP}{BiGaBP}{bilinear Gaussian belief propagation}
\newacronym{BP}{BP}{belief propagation}
\newacronym{BS}{BS}{base station}
\newacronym{BLER}{BLER}{block error rate}
\newacronym{CAP}{CAP}{central AP}
\newacronym{CCU}{CCU}{central computing unit}
\newacronym{CDF}{CDF}{cumulative distribution function}
\newacronym{CE}{CE}{channel estimation}
\newacronym{CFO}{CFO}{carrier frequency offset}
\newacronym{CIR}{CIR}{channel impulse response}
\newacronym{CLT}{CLT}{central limit theorem}
\newacronym{CP}{CP}{cyclic prefix}
\newacronym{CPU}{CPU}{central processing unit}
\newacronym{CRLB}{CRLB}{Cram\'{e}r--Rao lower bound}
\newacronym{CSI}{CSI}{channel state information}
\newacronym{CSIDCO}{CSIDCO}{complex SIDCO}
\newacronym{CSt-SBL}{CSt-SBL}{complex \textit{t}-distribution-based SBL}
\newacronym{DL}{DL}{deep learning}
\newacronym{DNN}{DNN}{deep neural network}
\newacronym{DSP}{DSP}{digital signal processor}
\newacronym{DoA}{DoA}{direction of arrival}
\newacronym{eMBB}{eMBB}{enhanced mobile broadband}
\newacronym{ECF}{ECF}{estimate-compress-forward}
\newacronym{EP}{EP}{expectation propagation}
\newacronym{FA}{FA}{false alarm}
\newacronym{FFT}{FFT}{fast Fourier transform}
\newacronym{FFNN}{FFNN}{feed-forward neural network}
\newacronym{FN}{FN}{factor node}
\newacronym{FG}{FG}{factor graph}
\newacronym{GaBP}{GaBP}{Gaussian belief propagation}
\newacronym{GM}{GM}{Gaussian-mixture}
\newacronym{IC}{IC}{interference cancellation}
\newacronym{IDD}{IDD}{iterative detection and decoding}
\newacronym{IFFT}{IFFT}{inverse fast Fourier transform}
\newacronym{i.i.d.}{i.i.d.}{independent and identically distributed}
\newacronym{ICI}{ICI}{inter-sub-frequency interference}
\newacronym{ISAC}{ISAC}{integrated sensing and communication}
\newacronym{JACDE}{JACDE}{joint activity, channel and data estimation}
\newacronym{JACE}{JACE}{joint activity and channel estimation}
\newacronym{JCDE}{JCDE}{joint channel and data estimation}
\newacronym{JCCE}{JCCE}{joint channel and CFO estimation}
\newacronym{JCCDE}{JCCDE}{joint channel, CFO, and data estimation}
\newacronym{KL}{KL}{Kullback-Leibler}
\newacronym{LAMP}{LAMP}{learned AMP}
\newacronym{LSA}{LSA}{latent semantic analysis}
\newacronym{LoS}{LoS}{line-of-sight}
\newacronym{LLR}{LLR}{log-likelihood ratio}
\newacronym{LMMSE}{LMMSE}{linear minimum mean square error}
\newacronym{LASSO}{LASSO}{least absolute shrinkage and selection operator}
\newacronym{MAC}{MAC}{multiple-access channel}
\newacronym{MAE}{MAE}{mean absolute error}
\newacronym{MAP}{MAP}{maximum \textit{a posteriori}}
\newacronym{MCS}{MCS}{modulation and coding scheme}
\newacronym{MPDQ}{MPDQ}{message passing DQ}
\newacronym{MD}{MD}{miss-detection}
\newacronym{MF}{MF}{matched filter}
\newacronym{MFB}{MFB}{matched filter bound}
\newacronym{MNS}{MNS}{minimum norm solution}
\newacronym{MI}{MI}{mutual information}
\newacronym{mMIMO}{mMIMO}{massive multiple-input multiple-output}
\newacronym{MIMO}{MIMO}{multiple-input multiple-output}
\newacronym{MIMO-OFDM}{MIMO-OFDM}{multiple-input multiple-output orthogonal frequency-division multiplexing}
\newacronym{MU-MIMO}{MU-MIMO}{multi-user multiple-input multiple-output}
\newacronym{MU-MIMO-OFDM}{MU-MIMO-OFDM}{multi-user multiple-input multiple-output orthogonal frequency-division multiplexing}
\newacronym{OTFS}{OTFS}{orthogonal time frequency space}
\newacronym{MIMO-OTFS}{MIMO-OTFS}{multiple-input multiple-output orthogonal time frequency space}
\newacronym{OTFS-CE}{OTFS-CE}{orthogonal time frequency space channel estimation}
\newacronym{MU-MIMO-OTFS}{MU-MIMO-OTFS}{multi-user multiple-input multiple-output orthogonal time frequency space}
\newacronym{DD}{DD}{delay--Doppler}
\newacronym{FT}{FT}{frequency--time}
\newacronym{MSP}{MSP}{modified subspace pursuit}
\newacronym{OMP}{OMP}{orthogonal matching pursuit}
\newacronym{ML}{ML}{maximum likelihood}
\newacronym{MMSE}{MMSE}{minimum mean square error}
\newacronym{MMV-AMP}{MMV-AMP}{multiple measurement vector approximate message passing}
\newacronym{MMV}{MMV}{multiple measurement vector}
\newacronym{MSE}{MSE}{mean square error}
\newacronym{MP}{MP}{message passing}
\newacronym{MPA}{MPA}{message passing algorithm}
\newacronym{MRC}{MRC}{maximal ratio combining}
\newacronym{MUD}{MUD}{multi-user detection}
\newacronym{mMTC}{mMTC}{massive machine type communications}
\newacronym{mmWave}{mmWave}{millimeter-wave}
\newacronym{NR}{NR}{new radio}
\newacronym{NMSE}{NMSE}{normalized mean square error}
\newacronym{OFDM}{OFDM}{orthogonal frequency-division multiplexing}
\newacronym{OLLA}{OLLA}{outer loop link adaptation}
\newacronym{PBI}{PBI}{parametric bilinear inference}
\newacronym{PDA}{PDA}{probabilistic data association}
\newacronym{PDF}{PDF}{probability density function}
\newacronym{PE}{PE}{prediction error}
\newacronym{PMF}{PMF}{probability mass function}
\newacronym{PN}{PN}{phase noise}
\newacronym{PPP}{PPP}{Poisson point process}
\newacronym{PSK}{PSK}{phase-shift keying}
\newacronym{QP}{QP}{quadratic programming}
\newacronym{QPSK}{QPSK}{quadrature PSK}
\newacronym{QAM}{QAM}{quadrature amplitude modulation}
\newacronym{RB}{RB}{resource block}
\newacronym{ReLU}{ReLU}{rectified linear unit}
\newacronym{RMSE}{RMSE}{root mean squared error}
\newacronym{RF}{RF}{radio frequency}
\newacronym{RX}{RX}{receive}
\newacronym{SAGE}{SAGE}{space-alternating generalized expectation-maximization}
\newacronym{SBL}{SBL}{sparse Bayesian learning}
\newacronym{SD}{SD}{sphere decoding}
\newacronym{SIDCO}{SIDCO}{sequential iterative decorrelation via convex optimization}
\newacronym{SGA}{SGA}{scalar Gaussian approximation}
\newacronym{SGD}{SGD}{stochastic gradient descent}
\newacronym{S-GAMP}{S-GAMP}{structured generalized approximate message passing}
\newacronym{SIC}{SIC}{soft interference cancellation}
\newacronym{SID}{SID}{self-iterative detection}
\newacronym{SIMO}{SIMO}{single-input multiple-output}
\newacronym{SINR}{SINR}{signal-to-interference-plus-noise power ratio}
\newacronym{SNR}{SNR}{signal-to-noise ratio}
\newacronym{Soft IC}{Soft IC}{soft interference cancellation}
\newacronym{SotA}{SotA}{state-of-the-art}
\newacronym{SPA}{SPA}{sum-product algorithm}
\newacronym{SSR}{SSR}{sparse signal recovery}
\newacronym{SVD}{SVD}{singular value decomposition}
\newacronym{TB}{TB}{transport block}
\newacronym{TDL}{TDL}{tapped delay line}
\newacronym{T-GaBP}{T-GaBP}{trainable GaBP}
\newacronym{T-GAMP}{T-GAMP}{trainable GAMP}
\newacronym{TX}{TX}{transmit}
\newacronym{UE}{UE}{user equipment}
\newacronym{ULA}{ULA}{uniform linear array}
\newacronym{URA}{URA}{uniform rectangular array}
\newacronym{URLLC}{URLLC}{ultra reliable low latency communications}
\newacronym{VAMP}{VAMP}{vector AMP}
\newacronym{VDB}{VDB}{vector database}
\newacronym{VGA}{VGA}{vector Gaussian approximation}
\newacronym{VN}{VN}{variable node}
\newacronym{VSS}{VSS}{vector similarity search}
\newacronym{ZF}{ZF}{zero-forcing}
\newacronym{flops}{flops}{floating point operations}
\newacronym{CS}{CS}{compressed sensing}
\newacronym{LIP}{LIP}{linear inference problem}
\newacronym{w.r.t.}{w.r.t.}{with respect to}
\newacronym{ISI}{ISI}{internal symbol interference}
\newacronym{SE}{SE}{spectral efficiency}
\newacronym{DU}{DU}{deep unfolding}
\newacronym{SISO}{SISO}{single-input single-output}
\newacronym{SISO-OTFS}{SISO-OTFS}{single-input single-output orthogonal time frequency space}
\newacronym{RG}{RG}{row and group}
\newacronym{RVM}{RVM}{relevance vector machine}
\newacronym{Ct-SBL}{Ct-SBL}{complex \textit{t}-distribution-based SBL}
\newacronym{MM}{MM}{majorization-minimization}
\newacronym{SMV}{SMV}{single measurement vector}
\newacronym{BCS}{BCS}{Bayesian compressed sensing}
\newacronym{UAV}{UAV}{unmanned aerial vehicle}
\newacronym{DCS}{DCS}{deterministic compressed sensing}
\newacronym{CDSC}{CDSC}{continuous-Doppler-spread channel}
\newacronym{ECR}{ECR}{equivalent channel response}
\newacronym{FISTA}{FISTA}{fast iterative shrinkage thresholding algorithm}
\newacronym{SBEM}{SBEM}{spatial basis expansion model}
\newacronym{VCR}{VCR}{virtual channel representation}
\newacronym{BDCPM}{BDCPM}{beam domain channel power matrix}
\newacronym{BSCM}{BSCM}{beam-based statistical channel model}
\newacronym{ELBO}{ELBO}{relaxed evidence lower bound}
\newacronym{V2X}{V2X}{Vehicle to everything}
\newacronym{LDSC}{LDSC}{limited-Doppler-spread channel}
\newacronym{HBM}{HBM}{hierarchical Bayesian model}
\newacronym{DT}{DT}{digital twin}
\newacronym{EM}{EM}{electromagnetic}
\newacronym{ROI}{ROI}{region of interest}
\newacronym{CPR}{CPR}{constitutive parameter reconstruction}
\newacronym{LSM}{LSM}{linear sampling method}
\newacronym{EII}{EII}{electromagnetic inverse imaging}
\newacronym{FDTD}{FDTD}{finite-difference time-domain}
\newacronym{LS}{LS}{Lippmann-Schwinger}
\newacronym{BIM}{BIM}{Born iterative method}
\newacronym{ASR}{ASR}{actual scatterer region}
\newacronym{NCC}{NCC}{normalized cross-correlation}
\newacronym{UCA}{UCA}{uniform circular antenna array}
\newacronym{WSS}{WSS}{wide-sense stationary}
\newacronym{FIM}{FIM}{Fisher information matrix}
\newacronym{3D}{3D}{three-dimensional}
\newacronym{IoT}{IoT}{Internet of Things}

\begin{document}

\title{
Ill-Posedness Analysis of CSI-Based Electromagnetic Inverse Scattering for Material Reconstruction in ISAC Systems
}

\author{ 
   Yubin Luo, Li Yu, \IEEEmembership{Member, IEEE}, Takumi Takahashi, \IEEEmembership{Member, IEEE}, Shaoyi Liu, Yuxiang Zhang, \IEEEmembership{Member, IEEE}, Jianhua Zhang, \IEEEmembership{Fellow, IEEE}, and Hideki Ochiai, \IEEEmembership{Fellow, IEEE}
    \thanks{Y. Luo, T. Takahashi and H. Ochiai are with the Graduate School of Engineering, The University of Osaka, 2-1-Yamada-oka, Suita, 565-0871, Japan. (e-mail: yubinluo0@gmail.com, takahashi@comm.eng.osaka-u.ac.jp, ochiai@comm.eng.osaka-u.ac.jp.)}
    \thanks{L. Yu, S. Liu, Y. Zhang, and J. Zhang are with the State Key Laboratory of Networking and Switching Technology, Beijing University of Posts and Telecommunications, Beijing 100876, China. (e-mail: li.yu@bupt.edu.cn; sy\_liu@bupt.edu.cn; zhangyx@bupt.edu.cn; jhzhang@bupt.edu.cn.)}
}


\maketitle

\begin{abstract}
\Ac{CSI}-based electromagnetic inverse scattering for material reconstruction in \ac{ISAC} systems offers a physics-grounded route to radio-environment awareness.
However, the \ac{CSI}-induced sensing operator is often ill-conditioned. 
This paper develops an operator-centric analysis by exposing the factorized \ac{ISAC} scattering matrix shaped by in-domain scattering responses and transceiver-side propagation channels. 
We identify the dominant source of near-rank deficiency as the background-related columns, which are highly coherent and nearly linearly dependent.
We further show that scatterer-related columns remain informative, whereas scatterer-column coherence decreases as the number of probing frequencies increases, thereby improving the effective rank of the sensing operator. 
A Hilbert-space interpretation generalizes this matrix-level insight and attributes the background--scatterer contrast to the structural difference between propagation-only and
material-induced scattering responses.
Building on this characterization, we derive a closed-form relation linking \ac{ROI} restriction, \ac{ROI} mismatch, effective coherence, and condition-number reduction, which explains why \ac{ROI}-restricted inversion stabilizes \ac{CSI}-based material reconstruction.
Guided by the analysis, a \ac{LSM}-initialized \ac{ROI}-constrained \ac{QP} formulation validates the predicted conditioning gains in a reduced subspace.
Finite-difference time-domain simulations over different target geometries and \acp{SNR} confirm the theoretical trends, including improved conditioning, reduced complexity, and robustness against representative inverse-scattering baselines.
\end{abstract}

\vspace{-3mm}

 \ifCLASSOPTIONpeerreview
 \begin{center} \bfseries EDICS Category: 3-BBND \end{center}
 \fi
\glsresetall
\IEEEpeerreviewmaketitle

\section{Introduction}

\Ac{ISAC} enables sensing and data transmission over the same waveform, spectrum, and hardware platform~\cite{b1}.
By reusing in-band communication signals and \ac{RF} chains, it provides time-synchronized and periodic observations with high spectral efficiency and low system cost~\cite{b2,b3}.
These capabilities make \ac{ISAC} a key enabling technology for \ac{6G}, supporting diverse applications such as multi-device \ac{IoT}, vehicular networks, cooperative access, environment-aware beam management, and wireless \acp{DT}~\cite{b4,b46}.

A common requirement across these \ac{ISAC} applications is reliable radio-environment awareness.
Accurate channel modeling is therefore a fundamental prerequisite for environment-aware \ac{ISAC} operation.
Deterministic channel modeling, such as ray tracing and full-wave simulation, provides a physics-consistent way to represent such environments~\cite{b5,b6}.
Beyond multipath geometry, these models crucially depend on the constitutive parameters of dominant scatterers, which directly govern the electromagnetic channel response.
Accordingly, \ac{CPR} becomes an essential physical-layer enabler for environment-aware \ac{ISAC} applications in \ac{6G} systems.

A variety of \ac{CPR} methodologies have been investigated in science and
engineering. On one end, direct \ac{EM} characterization with specialized
instruments, such as waveguides and resonators, can achieve high accuracy
but typically requires dedicated setups and offline measurements
~\cite{b8,b9,b10}. On the other end, multimodal and data-driven approaches
leverage auxiliary sensing modalities and training data to infer material
parameters at scale~\cite{b11,b12,b13,b14,b15,b16}. While effective in
controlled settings, these approaches are not always compatible with
large-scale \ac{6G} deployment, where wireless systems are expected to rely
on communication-native and periodically updated observations. This
motivates a physics-grounded route, namely \ac{EII}, which formulates
\ac{CPR} as an inverse scattering problem governed by Maxwell's equations
and integral operators~\cite{b17,b18,b19}. 

In the frequency domain, the wireless multipath channel can be viewed as a parametric surrogate of Green's-function-based propagation~\cite{b20}. Therefore, in \ac{ISAC} systems, our forward model starts from the communication link and treats the resulting \ac{CSI} as an intrinsic sensing observation. It serves as a surrogate for conventional scattered-field measurements and yields a physics-based mapping from constitutive parameters to the channel observations~\cite{b21,b22,b23}. In this work, we focus on the sensing side of \ac{ISAC}, where the environment is inferred from communication-generated observations, and aim to provide an operator-centric theoretical characterization of ill-posedness under this pilot-aided \ac{CSI}-driven observation mechanism. Fig.~\ref{fig:isac} illustrates the considered ISAC-enabled sensing scenario.

\begin{figure}[t]
  \centering
  \includegraphics[width=0.9\columnwidth]{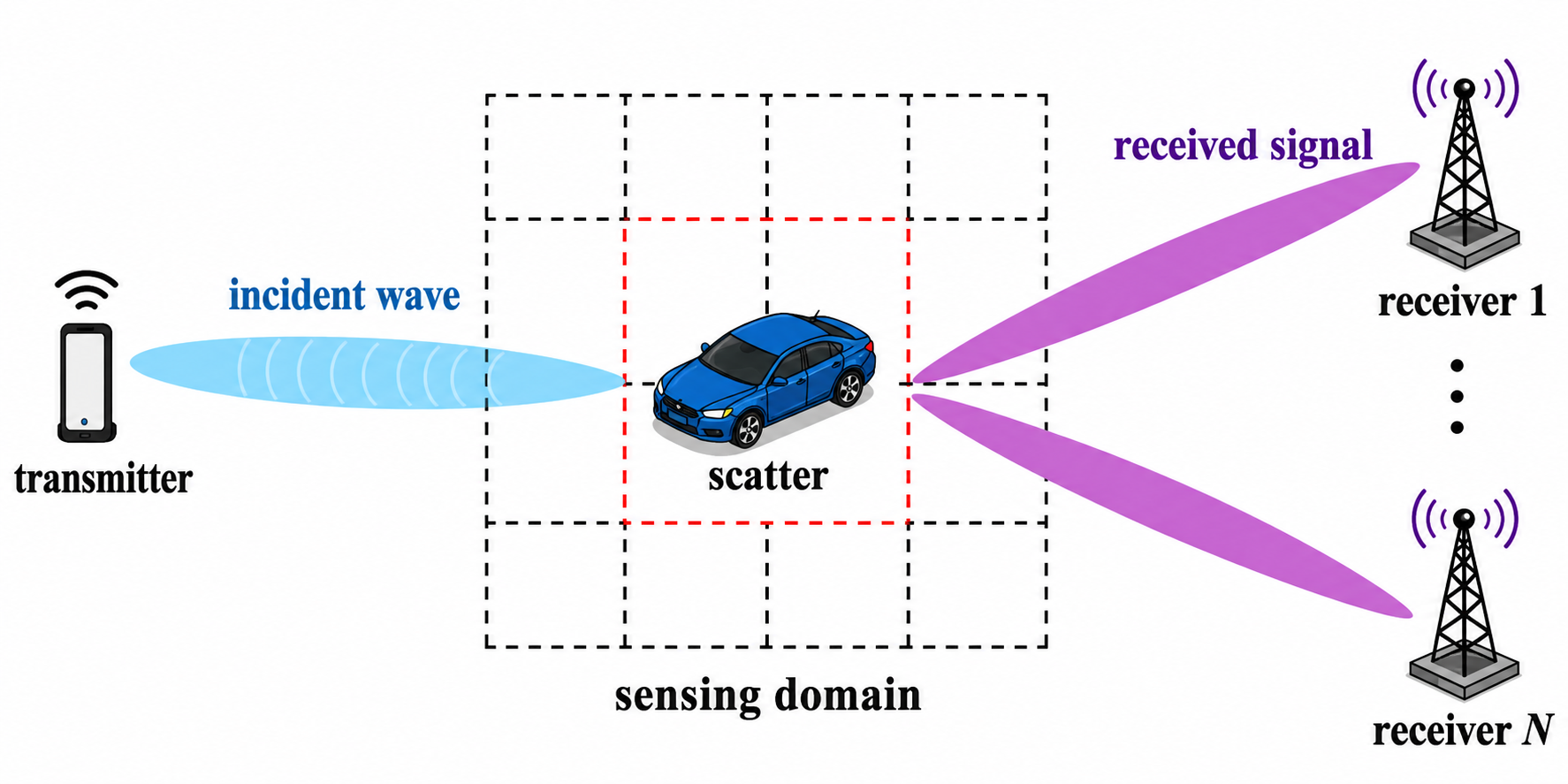}
  \caption{Schematic diagram of the ISAC system.}
  \label{fig:isac}
  \vspace{3mm}
\end{figure}

\ac{ISAC}-enabled \ac{EII} still faces structural limitations, and we highlight two closely related issues.
First, \ac{DT}-oriented deterministic channel modeling demands high-precision \ac{CPR}, yet inverse-scattering-based \ac{CPR} is inherently often ill-conditioned.
Small singular values of the forward scattering matrix that maps the scene contrast to the measured \ac{CSI} indicate severe ill-conditioning, leading to high sensitivity to measurement perturbations and pronounced noise amplification~\cite{b24,b25}.
Second, migrating from conventional scattered-field measurements to \ac{ISAC}-generated \ac{CSI} fundamentally reshapes the forward model. Unlike classical \ac{EII} settings, the \ac{CSI} is jointly shaped by waveform-dependent pilots and the transmitter- and receiver-side propagation channels, resulting in a structured and factorized scattering operator.

The importance of mitigating ill-conditioning has been well recognized under classical scattered-field measurements, and a broad set of techniques---including regularization, multi-frequency, and multi-view strategies---have been developed~\cite{b12,b13}.
\Ac{ROI} constraints and sampling-based methods have also been employed as practical means to reduce the size of the problem space, thereby enhancing robustness~\cite{b45}.
Nevertheless, it remains unclear how the \ac{ISAC}-induced scattering operator alters the ill-posedness analysis in \ac{CSI}-driven \ac{CPR}, and a unified operator-centric characterization is still limited.
In particular, quantitative operator-level studies tailored to \ac{ISAC} scenarios remain limited. Existing works do not always (i) pinpoint the dominant source of ill-conditioning, (ii) convert the analysis into provable conditioning gains, or (iii) leverage these insights to guide targeted algorithm design that mitigates ill-conditioning in \ac{CPR}.
Moreover, \Ac{ROI} selection is an intrinsic modeling choice and, when mismatched, can affect both reconstruction accuracy and computational complexity.

Motivated by these gaps, we investigate \ac{CSI}-based \ac{CPR} in \ac{ISAC} systems under a \ac{MIMO} setting and characterize how the structure of the \ac{ISAC}-induced operator governs the ill-posedness of the associated inverse scattering problem.
Our main contributions are as follows:
\begin{itemize}
\item

We characterize the electromagnetic measurement operator in \ac{ISAC} systems by revealing its channel-shaped factorization into propagation channels and a scattering-response operator.
Building on this structure, we analyze an operator-centric ill-posedness: the systematic ill-conditioning is dominated by highly coherent columns of the scattering matrix associated with the background (air) region, whereas columns corresponding to the actual scatterer region remain comparatively informative due to their weak inter-column coherence, which further decreases with the number of probing frequencies.
Based on this analysis, we show that \ac{ROI} restriction can alleviate ill-posedness by suppressing the redundancy-dominant background subspace.

\item
We extend the discrete column-coherence analysis to a Hilbert-space operator interpretation, revealing that the background--\ac{ASR} distinction stems from the different mechanisms of propagation-only background responses and material-induced scattering responses. 
This extension shows that the proposed ill-posedness mechanism is not tied to a particular discretization but persists at the operator level. We further derive a closed-form condition-number relation that connects \ac{ROI} restriction, \ac{ROI} mismatch, and effective coherence.
\item
To operationalize and empirically validate the above operator-level analysis, we use an ROI-restricted \ac{QP} update as a numerical vehicle: the \ac{LSM} yields a coarse support estimate to define the \ac{ROI}, and the subsequent QP is solved only on the reduced subspace.
\end{itemize}

The remainder of this paper is organized as follows.
Section~II introduces the \ac{CSI}-based forward model and formulates the corresponding \ac{CPR} problem.
Section~III provides an operator-centric characterization by revealing the channel-shaped factorization of the forward scattering operator and identifying the dominant mechanisms behind ill-conditioning.
Section~IV presents an \ac{ROI}-restricted reconstruction framework to mitigate ill-posedness, together with a conditioning analysis that relates \ac{ROI} restriction.
Section~V reports numerical results, and Section~VI concludes the paper.

\textbf{Notation:} Throughout the paper, \textbf{boldface} symbols denote vectors or matrices; 
\(\mathrm{j}\) represents the imaginary unit; 
\((\cdot)^{\mathrm{H}}\) and \((\cdot)^{\mathrm{T}}\) denote the Hermitian (conjugate transpose) and transpose operations, respectively; $(\cdot)^{*}$ denotes the complex conjugate of a scalar; \(\langle\cdot,\cdot\rangle\) denotes the inner product; $\|\cdot\|_F$ denotes the Frobenius norm; 
\(\|\cdot\|_2\) denotes the Euclidean norm; \(\circ\) stands for the Khatri--Rao product; \(\mathrm{diag}(\mathbf{x})\) denotes a diagonal matrix whose diagonal entries are formed from the elements of vector \(\mathbf{x}\); 
\(\mathbf{A}_{[:,i]}\) indicates the \(i\)-th column of matrix \(\mathbf{A}\); 
\(\mathbf{I}\) denotes the identity matrix; 
\(\mathbf{A} \succeq 0\) indicates that matrix \(\mathbf{A}\) is positive semi-definite.

%\begin{comment}

%\end{comment}

\vspace{-0.5em}
\section{Basic System Model and Problem Formulation}

In this section, we start from a standard baseband channel model and then derive a physics-consistent closed-form expression of the corresponding \ac{EM} scattering channel from electromagnetic field theory~\cite{b21}.
To reduce the computational complexity of both the forward and inverse problems and to enable efficient repeated solutions, we adopt a two-dimensional transverse magnetic with respect to the $z$-axis (TM$_z$)  with time dependence $e^{\mathrm{j}\omega t}$~\cite{b31}.
Let $D\subset\mathbb{R}^2$ denote the sensing domain and let $\Gamma$ denote the observation curve.

Consider a point-to-point \ac{MIMO} link with an $N_t$-element Tx array and an $N_r$-element Rx array.
The proposed input--output formulation is also applicable to distributed deployments, in which the $N_t$ transmitters and $N_r$ receivers are geographically separated single-antenna nodes that are synchronized in time and frequency. Additionally, we assume that all antenna elements are located on an observation curve $\Gamma$ in the far field of the imaging domain $D$.

The system operates over $K$ probing frequencies:
\begin{equation}
f_k \triangleq f_c + \left(k-\frac{K+1}{2}\right)\Delta f,\qquad k=1,\ldots,K,
\label{eq:freq_grid}
\end{equation}
where $f_c$ and $\Delta f$ denote the center frequency and the frequency spacing, respectively.
At each frequency $f_k$, the Tx sends $T$ pilot symbols (time slots) to probe the channel.
In slot $t$, a known pilot excitation $\mathbf x_{k,t}\in\mathbb C^{N_t\times 1}$ produces the received observation
$\mathbf y_{k,t}\in\mathbb C^{N_r\times 1}$ according to:
\begin{equation}
\mathbf y_{k,t}=\mathbf H_k\,\mathbf x_{k,t}+\mathbf n_{k,t},
\qquad
\mathbf H_k\in\mathbb C^{N_r\times N_t},
\qquad
\mathbf n_{k,t}\in\mathbb C^{N_r\times 1},
\label{eq:pilot_io_per_slot}
\end{equation}
where $\mathbf H_k$ denotes the end-to-end channel matrix at tone $k$.
Our goal is to express $\mathbf H_k$ in closed form as a function of the unknown material contrast so that inverse-scattering-based \ac{CPR} can be carried out directly from the pilot-based \ac{CSI}.

\vspace{-1.2em}
\subsection{End-to-End Scattering Channel}
For the \(k\)-th frequency, with angular frequency
\(\omega_k=2\pi f_k\), the total field inside the investigation domain
\(D\) satisfies the \ac{LS} integral equation~\cite{b26,b27,b28}
\begin{equation}
E_k^{t}(\mathbf r_d)
=
E_k^{i}(\mathbf r_d)
+
\int_D
G_k(\mathbf r_d,\mathbf r')\,
\chi(\mathbf r';\omega_k)\,
E_k^{t}(\mathbf r')\,\mathrm d\mathbf r',
\label{eq:et}
\end{equation}
where \(E_k^{t}\) and \(E_k^{i}\) denote the total and incident fields,
respectively, \(\mathbf r_d\in D\) is the investigated point, and
\(\mathbf r'\in D\) is the integration point. The two-dimensional Green's
function is
\begin{equation}
G_k(\mathbf r,\mathbf r')
\triangleq
-\frac{\mathrm j}{4} k_b^2(\omega_k)
H_0^{(2)}\!\left(k_b(\omega_k)\|\mathbf r-\mathbf r'\|\right),
\end{equation}
where \(H_0^{(2)}(\cdot)\) is the zeroth-order Hankel function of the
second kind, \(k_b(\omega_k)=\omega_k/c\), and \(c\) is the speed of
light~\cite{b29,b30}. The material contrast is defined as
\begin{equation}
\chi(\mathbf r;\omega_k)
\triangleq
\varepsilon(\mathbf r)-1
+
\mathrm j
\frac{\sigma(\mathbf r)}{\varepsilon_0\omega_k},
\label{eq:contrast}
\end{equation}
where \(\varepsilon(\mathbf r)\) is the relative permittivity,
\(\sigma(\mathbf r)\) is the conductivity, and \(\varepsilon_0\) is the
vacuum permittivity. Since the material parameters are assumed to vary
slowly over the operating band, we use
\(\chi(\mathbf r;\omega_k)\approx\chi(\mathbf r;\omega_c)\) for all \(k\).

After discretizing \(D\) into \(N\) sampling points, \eqref{eq:et}
becomes
\begin{equation}
\mathbf E_k^{t,D}
=
\mathbf E_k^{i,D}
+
\mathbf G_k^D
\operatorname{diag}(\boldsymbol\chi)
\mathbf E_k^{t,D},
\label{eq:et_matrix}
\end{equation}
where
\(\mathbf E_k^{t,D},\mathbf E_k^{i,D}\in\mathbb C^{N\times N_t}\)
are the total and incident field matrices over \(D\),
\(\boldsymbol\chi\in\mathbb C^{N\times 1}\) is the discretized contrast
vector, and \(\mathbf G_k^D\in\mathbb C^{N\times N}\) is the in-domain
Green's matrix. The induced scattered field inside \(D\) is denoted by
\(\mathbf E_k^{s,D}\in\mathbb C^{N\times N_t}\), with
\(\mathbf E_k^{t,D}=\mathbf E_k^{i,D}+\mathbf E_k^{s,D}\).

Let \(\mathbf H_{1,k}\in\mathbb C^{N\times N_t}\) and
\(\mathbf H_{2,k}\in\mathbb C^{N_r\times N}\) denote the propagation
operators from the transmitter to \(D\) and from \(D\) to the receiver
array on \(\Gamma\), respectively. For the \(t\)-th pilot slot at
frequency \(f_k\), the transmit pilot is
\(\mathbf x_{k,t}\in\mathbb C^{N_t\times 1}\), and the incident field
over \(D\) is
\begin{equation}
\mathbf E_{k,t}^{i,D}
=
\mathbf H_{1,k}\mathbf x_{k,t}.
\end{equation}
Combining \eqref{eq:et_matrix} with the propagation from \(D\) to
\(\Gamma\), the scattered field observed at the receiver side is
\begin{equation}
\mathbf E_{k,t}^{s,\Gamma}
=
\mathbf H_{2,k}
\operatorname{diag}(\boldsymbol\chi)
\Big(
\mathbf I_N
-
\mathbf G_k^D\operatorname{diag}(\boldsymbol\chi)
\Big)^{-1}
\mathbf H_{1,k}\mathbf x_{k,t},
\label{eq:channel}
\end{equation}
where
\(\mathbf E_{k,t}^{s,\Gamma}\in\mathbb C^{N_r\times 1}\) denotes the
scattered field on \(\Gamma\)~\cite{b21}. In free space,
\(\mathbf H_{1,k}\) and \(\mathbf H_{2,k}\) reduce to the corresponding
Green's-function propagation matrices. Throughout this paper, we consider
nominal coherently calibrated and phase-tracked propagation operators,
which are generated by simulation. This is consistent with coherent
\ac{MIMO} imaging, where array calibration and pilot-aided phase-noise
tracking are commonly adopted~\cite{b34,b35,b36}.

The received observation at the \(t\)-th pilot slot is therefore written as
\begin{equation}
\mathbf y_{k,t}
=
\mathbf E_{k,t}^{s,\Gamma}
+
\mathbf n_{k,t},
\qquad
\mathbf y_{k,t},\mathbf n_{k,t}\in\mathbb C^{N_r\times 1},
\label{eq:pilot_esca_per_slot}
\end{equation}
where \(\mathbf n_{k,t}\) denotes receiver noise. Stacking
\(\{\mathbf y_{k,t}\}_{t=1}^{T}\) gives
\begin{equation}
\mathbf y_k
\triangleq
\big[
\mathbf y_{k,1}^{\mathrm T},
\ldots,
\mathbf y_{k,T}^{\mathrm T}
\big]^{\mathrm T}
=
\mathbf A_k\boldsymbol\chi+\mathbf n_k,
\label{eq:yk}
\end{equation}
where
\(\mathbf y_k,\mathbf n_k\in\mathbb C^{TN_r\times 1}\). Defining
\(\mathbf X_k\triangleq
[\mathbf x_{k,1},\ldots,\mathbf x_{k,T}]
\in\mathbb C^{N_t\times T}\), the forward operator
\(\mathbf A_k\in\mathbb C^{TN_r\times N}\) is given by~\cite{b21}
\begin{equation}
\mathbf A_k
\triangleq
\left(
\mathbf X_k^{\mathrm T}
\mathbf H_{1,k}^{\mathrm T}
\left[
\left(
\mathbf I_N
-
\mathbf G_k^D\operatorname{diag}(\boldsymbol\chi)
\right)^{-\mathrm T}
\right]
\right)
\circ
\mathbf H_{2,k},
\label{eq:Ak}
\end{equation}
where \(\circ\) denotes the column-wise Khatri--Rao product, which maps
the factors of size \(T\times N\) and \(N_r\times N\) to
\(\mathbb C^{TN_r\times N}\).

Since \(\mathbf A_k\) depends on \(\boldsymbol\chi\), \eqref{eq:yk} is
nonlinear. We solve this problem using the \ac{BIM}. At the first
iteration, the Born approximation
\(\mathbf E_k^{t,D}\approx\mathbf E_k^{i,D}\) yields
\begin{equation}
\mathbf y_k
\approx
\mathbf A_k^{(1)}\boldsymbol\chi+\mathbf n_k,
\qquad
\mathbf A_k^{(1)}
=
\big(
\mathbf X_k^{\mathrm T}\mathbf H_{1,k}^{\mathrm T}
\big)
\circ
\mathbf H_{2,k}.
\label{eq:born_first_operator}
\end{equation}
For iteration \(n\ge2\), \(\mathbf A_k^{(n)}\) is updated by substituting
the previous estimate \(\boldsymbol\chi^{(n-1)}\) into \eqref{eq:Ak}.
The next estimate \(\boldsymbol\chi^{(n)}\) is then obtained from the
regularized linear inverse problem
\(\mathbf y_k\approx\mathbf A_k^{(n)}\boldsymbol\chi+\mathbf n_k\)
until convergence~\cite{b32,b33}. The detailed reconstruction algorithm
is presented in Section~IV.

\vspace{-0.8em}
\subsection{Linear Sampling Method}
We briefly recall the \ac{LSM}~\cite{b37, b41}, which is used only to obtain
a coarse scatterer support and initialize the \ac{ROI} for the subsequent
Born-iterative quantitative reconstruction.

For frequency $f_k$, the stacked measurement
$\mathbf y_k\in\mathbb C^{TN_r\times 1}$ is given by \eqref{eq:yk}.
Under the quasi-static assumption within one sensing snapshot, the receive
response induced by a unit excitation at the $m$-th transmitter is denoted
by $\mathbf h_{k,m}\in\mathbb C^{N_r\times 1}$. The $t$-th pilot-slot
observation can then be expressed as
\begin{equation}
\mathbf y_{k,t}
=
\sum_{m=1}^{N_t}
x_{k,t}(m)\mathbf h_{k,m}
+
\mathbf n_{k,t},
\label{eq:yk_t_superpos_lsm}
\end{equation}
where $x_{k,t}(m)$ is the $m$-th entry of $\mathbf x_{k,t}$.

To match the stacking of $\mathbf y_k$, we define the slot-stacked
response
\begin{equation}
\mathbf h_{k,m}^{\mathrm{stk}}
\triangleq
\big[
x_{k,1}(m)\mathbf h_{k,m}^{\mathrm T},
\ldots,
x_{k,T}(m)\mathbf h_{k,m}^{\mathrm T}
\big]^{\mathrm T}
\in\mathbb C^{TN_r\times 1},
\label{eq:hkm_stack_lsm}
\end{equation}
where the pilot weights are inherited from the slot-wise observation in
\eqref{eq:yk_t_superpos_lsm}. The multi-static response matrix is then
\begin{equation}
\mathbf U_k
\triangleq
\big[
\mathbf h_{k,1}^{\mathrm{stk}},
\ldots,
\mathbf h_{k,N_t}^{\mathrm{stk}}
\big]
\in\mathbb C^{TN_r\times N_t}.
\label{eq:Uk_def_lsm}
\end{equation}

For each test point in $D$, \ac{LSM} seeks a coefficient vector whose
multi-static response matches the corresponding scatterer-to-Rx Green's
response. In matrix form, this is written as
\begin{equation}
\mathbf U_k\mathbf c_k
\approx
\mathbf G_{2,k}^{\Gamma},
\label{eq:lsm_sys}
\end{equation}
where $\mathbf c_k\in\mathbb C^{N_t\times N}$ is the \ac{LSM} coefficient
matrix, and $\mathbf G_{2,k}^{\Gamma}\in\mathbb C^{TN_r\times N}$ is the
time-slot-stacked scatterer-to-Rx Green's-function matrix. Since the
Green's function is invariant over the $T$ pilot slots within one snapshot,
$\mathbf G_{2,k}^{\Gamma}$ is formed by stacking the same receive Green's
response across the $T$ slots.

In practice, $\mathbf c_k$ is obtained by Tikhonov regularization~\cite{b40}:
\begin{equation}
\mathbf c_k
=
\arg\min_{\mathbf c}
\left\|
\mathbf U_k\mathbf c
-
\mathbf G_{2,k}^{\Gamma}
\right\|_F^2
+
\zeta\|\mathbf c\|_F^2,
\label{eq:lsm}
\end{equation}
where $\zeta>0$ is the regularization parameter. After computing
$\{\mathbf c_k\}_{k=1}^{K}$, the multi-frequency indicator is constructed
as~\cite{b41}
\begin{equation}
\mathcal J_K(\mathbf r)
=
\frac{1}{K}
\sum_{k=1}^{K}
\log_{10}
\left(
\|\mathbf c_k(\mathbf r)\|_2^2
\right),
\label{eq:lsm_indicator}
\end{equation}
where $\mathbf c_k(\mathbf r)\in\mathbb C^{N_t}$ denotes the column of
$\mathbf c_k$ associated with the test point $\mathbf r$. The resulting
indicator is subsequently thresholded to construct the \ac{ROI}, while the
quantitative contrast reconstruction is performed by the Born-iterative
method introduced in Section~IV.

\section{Ill-Posedness Analysis of the ISAC-CSI Operator}

\subsection{Operator Structure and Coherence Metric}

As discussed in Section~II-A, the estimation of material contrast from \ac{ISAC} observations leads to a linear inverse problem of the form~\eqref{eq:yk}.
In practice, this inverse problem is often severely ill-conditioned, which manifests as pronounced sensitivity to noise and modeling mismatch.
A standard quantitative measure of this ill-posedness is the condition number of the forward operator $\mathbf{A}_k$~\cite{b42}:
\begin{equation}
\kappa(\mathbf{A}_k)\triangleq\frac{\sigma_{\max}(\mathbf{A}_k)}{\sigma_{\min}(\mathbf{A}_k)},
\end{equation}
where $\sigma_{\max}(\mathbf{A}_k)$ and $\sigma_{\min}(\mathbf{A}_k)$ denote the largest and smallest singular values of $\mathbf{A}_k$, respectively.
In the ill-conditioned regime, the rapid decay of the singular values drives $\sigma_{\min}(\mathbf{A}_k) \approx 0$, rendering $\mathbf{A}_k$ nearly rank-deficient and the inverse problem highly unstable.

A key distinction is that $\mathbf{A}_k$ is channel-shaped: it is formed from \ac{CSI} extracted from the communication link rather than from direct scattered-field measurements.
Although the formulation ultimately traces back to the Lippmann--Schwinger equation, the intermediate \ac{CSI} extraction process fundamentally reshapes the forward operator.
Specifically, after pilot-assisted \ac{CSI} extraction with a fixed waveform, the segmented Tx--scatterer--Rx propagation inherent to the wireless channel induces a channelized operator structure, where each column not only encodes the sensitivity to a local contrast perturbation but is also predominantly shaped by propagation and in-domain scattering coupling.

This structure becomes explicit by examining individual columns of $\mathbf{A}_k$.
Let $\mathbf a_{k,j}$ denote the $j$-th column of $\mathbf A_k$.
Under the \ac{CSI}-based observation model, $\mathbf a_{k,j}$ takes a Khatri--Rao product form:
\begin{equation}
\mathbf a_{k,j} = (\mathbf A_k)_{[:,j]} = \mathbf u_j \circ \mathbf v_j,
\label{eq:col_decomp_main}
\end{equation}
where
\begin{equation}\label{eq:uv_def_main}
\begin{aligned}
\mathbf u_j &\triangleq
\Big(
\mathbf X_k^{\mathrm T}\mathbf H_{1,k}^{\mathrm T}
(\mathbf I-\mathbf G_k^{D}\mathrm{diag}(\boldsymbol\chi))^{-{\mathrm T}}
\Big)_{[:,j]},\\
\mathbf v_j &\triangleq (\mathbf H_{2,k})_{[:,j]}.
\end{aligned}
\end{equation}

Note that $\mathbf u_j\in\mathbb C^{T\times 1}$ is the slot-domain Tx-side signature associated with pixel $j$ after pilot mixing and in-domain coupling, while $\mathbf v_j\in\mathbb C^{N_r\times 1}$ captures the Rx-side propagation.

From a channel interpretation, $\mathbf u_j$ represents the pilot-induced total-field response associated with pixel $j$, encompassing transmitter-side waveform encoding, Tx-array-to-pixel propagation, and in-domain multiple-scattering coupling.
As a result, correlations among $\{\mathbf u_j\}$ are jointly shaped by the pilot structure, propagation characteristics, and the contrast distribution.
In contrast, $\{\mathbf v_j\}$ depends solely on pixel-to-receiver propagation and the receive-array manifold and therefore admits a simpler, geometry-driven interpretation.

The Khatri--Rao structure in~\eqref{eq:col_decomp_main} implies that the conditioning properties of $\mathbf A_k$ are governed by the joint correlation behavior of $\{\mathbf u_j\}$ and $\{\mathbf v_j\}$.
To quantify this effect, we characterize the inter-column dependence using the \ac{NCC}:
\begin{equation}
\mu_{ij} \triangleq
\frac{\big|\langle \mathbf a_{k,i},\,\mathbf a_{k,j}\rangle\big|}
     {\|\mathbf a_{k,i}\|_2\,\|\mathbf a_{k,j}\|_2} \in [0,1],
\label{eq:ncc_def}
\end{equation}
which serves as a direct proxy for column coherence and, consequently, for the degree of ill-conditioning.

In this work, we focus exclusively on the propagation-channel-induced structure of the ISAC scattering operator and accordingly fix the pilot configuration throughout the analysis. The impact of pilot design on CPR performance has been studied in~\cite{b22}.

With respect to propagation effects, we adopt a conventional channel model as a neutral baseline.
For each probing tone $k$, the channel is assumed quasi-static within one sensing snapshot and to exhibit approximately stationary spatial statistics over the aperture.
Across tones, we allow weak frequency coherence and only assume that the cross-tone correlation is small in a second-order sense.
These assumptions aim to capture generic propagation behavior without injecting site-specific structure that could bias the column-correlation analysis.
They are typically reasonable under standard coherence conditions: the quasi-static approximation holds when the snapshot duration is shorter than the channel coherence time, while the inter-tone correlation becomes weak when the tone spacing $\Delta f$ is comparable to or larger than the coherence bandwidth~\cite{b38,b39}.

Under this controlled setting, we show that the contrast between the background (air) and \ac{ASR} governs the inter-column correlations of the forward operator and ultimately drives the observed ill-posedness.

\vspace{-1em}
\subsection{Air-Region Columns: Strong Coherence}

\begin{lemma}
\label{lem:air_cols_dep}
Columns of $\mathbf A_k$ associated with the background (air) region are highly coherent and thus approximately linearly dependent.
\end{lemma}

\begin{proof}
From~\eqref{eq:uv_def_main}, when the sampling pixel $j$ lies in the background (air) region where the contrast vanishes, the pilot-induced total-field factor $\mathbf u_j$ reduces to
\begin{equation}
\mathbf u_j
= \big(\mathbf X_k^{\mathrm T}\mathbf H_{1,k}^{\mathrm T}\big)_{[:,j]}
= \big(\mathbf H_{1,k}\mathbf X_k\big)^{\mathrm T}_{[:,j]} .
\label{eq:degradation_u}
\end{equation}
In this case, the scatterer-induced response vanishes, and $\mathbf u_j$ is governed solely by the transmitter-to-pixel propagation channel.

Under the fixed pilot configuration assumed throughout this work, a spatially stationary propagation channel $\mathbf H_{1,k}$ induces pronounced correlations among the columns corresponding to nearby background pixels.
Furthermore, due to the Khatri--Rao structure $\mathbf a_{k,j}=\mathbf u_j\circ \mathbf v_j$, similar propagation-induced correlations in the pixel-to-receiver channel $\mathbf H_{2,k}$ are inherited by the columns of $\mathbf A_k$.

Consequently, in the absence of scatterer-induced interactions, the columns of $\mathbf A_k$ associated with the background region predominantly reflect the statistics of the underlying propagation channels.
This leads to high mutual coherence and, hence, approximate linear dependence among air-region columns.
A rigorous bound on the corresponding column coherence is provided in Appendix~B.
\end{proof}

\begin{remark}
The coherence analysis above is based on the nominal propagation operators $\mathbf H_{1,k}$ and $\mathbf H_{2,k}$, which are assumed to be coherently calibrated and phase-tracked.
In practical millimeter-wave \ac{ISAC} systems, this nominal model may be affected by hardware non-idealities such as oscillator phase noise, array mutual coupling, residual gain--phase mismatch, and calibration errors.
Although these effects have different physical origins, their first-order impact is to perturb the effective Tx-to-domain and domain-to-Rx propagation operators.
We therefore collect their aggregate influence into additive operator mismatches and write
\begin{equation}
\widetilde{\mathbf H}_{\ell,k}
=
\mathbf H_{\ell,k}
+
\Delta \mathbf H_{\ell,k},
\qquad \ell\in\{1,2\}.
\end{equation}

To examine whether such perturbations invalidate the background-column
coherence mechanism, define
\begin{equation}
\mathbf S_{\chi,k}
\triangleq
\left[
\mathbf I-\mathbf G_k^D\operatorname{diag}(\boldsymbol\chi)
\right]^{-\mathrm T},
\quad
\mathbf B_k
\triangleq
\mathbf X_k^{\mathrm T}\mathbf H_{1,k}^{\mathrm T}\mathbf S_{\chi,k}.
\end{equation}
Then the nominal sensing operator can be written as $\mathbf A_k=\mathbf B_k\circ\mathbf H_{2,k}$.
Under the perturbed propagation operators, the perturbed sensing operator is then
\begin{equation}
\widetilde{\mathbf A}_k
=
\widetilde{\mathbf B}_k\circ\widetilde{\mathbf H}_{2,k}
=
\mathbf A_k+\Delta\mathbf A_k,
\end{equation}
where $\Delta\mathbf A_k $ denotes the operator perturbation.
Using the Frobenius-norm bound for the Khatri--Rao product, the induced
operator perturbation satisfies
\begin{equation}
\begin{aligned}
\|\Delta \mathbf A_k\|_F
&\le
\|\Delta \mathbf B_k\circ \mathbf H_{2,k}\|_F
+
\|\mathbf B_k\circ \Delta \mathbf H_{2,k}\|_F  \\
&\quad
+
\|\Delta \mathbf B_k\circ \Delta \mathbf H_{2,k}\|_F  \\
&\le
\|\Delta \mathbf B_k\|_F\|\mathbf H_{2,k}\|_F
+
\|\mathbf B_k\|_F\|\Delta \mathbf H_{2,k}\|_F  \\
&\quad
+
\|\Delta \mathbf B_k\|_F\|\Delta \mathbf H_{2,k}\|_F,
\end{aligned}
\end{equation}
where $\Delta \mathbf B_k = \mathbf X_k^{\mathrm T} (\Delta \mathbf H_{1,k})^{\mathrm T} \mathbf S_{\chi,k}$.
Accordingly, we introduce an aggregate relative operator-mismatch bound
\begin{equation}
\frac{\|\Delta \mathbf A_k\|_F}{\|\mathbf A_k\|_F}
\le
\mathcal E_k,
\end{equation}
where $\mathcal E_k$ upper-bounds the relative sensing-operator error
induced by aggregate propagation mismatch at the $k$-th tone.

For the $j$-th column,
\begin{equation}
\widetilde{\mathbf a}_{k,j}
=
\mathbf a_{k,j}
+
\delta \mathbf a_{k,j},
\qquad
\|\delta \mathbf a_{k,j}\|_2
\le
\|\Delta \mathbf A_k\|_F.
\end{equation}
If the column energies are not severely imbalanced
\begin{equation}
\|\mathbf a_{k,j}\|_2
\ge
c_a\|\mathbf A_k\|_F,
\qquad c_a>0,
\end{equation}
where $c_a$ only excludes nearly vanishing columns and may scale with the number of discretization cells, then there exists $\epsilon_{a,k}>0$ such that
\begin{equation}
\frac{\|\delta \mathbf a_{k,j}\|_2}{\|\mathbf a_{k,j}\|_2}
\le
\epsilon_{a,k}
\le
\frac{\mathcal E_k}{c_a}.
\end{equation}

This relative column perturbation yields a lower bound on the perturbed column coherence. 
Applying the reverse triangle inequality to the numerator of $\widetilde{\mu}_{ij}$, and then using the Cauchy--Schwarz inequality with $\|\delta\mathbf a_{k,i}\|_2\le \epsilon_{a,k}\|\mathbf a_{k,i}\|_2$ and $\|\delta\mathbf a_{k,j}\|_2\le \epsilon_{a,k}\|\mathbf a_{k,j}\|_2$, gives
\begin{equation}
\left|
\left\langle
\widetilde{\mathbf a}_{i},
\widetilde{\mathbf a}_{j}
\right\rangle
\right|
\ge
\left(
\mu_{ij}
-
2\epsilon_{a,k}
-
\epsilon_{a,k}^{2}
\right)
\|\mathbf a_{k,i}\|_2\|\mathbf a_{k,j}\|_2.
\end{equation}
Moreover, the denominator satisfies
\begin{equation}
\|\widetilde{\mathbf a}_{k,i}\|_2
\|\widetilde{\mathbf a}_{k,j}\|_2
\le
(1+\epsilon_{a,k})^2
\|\mathbf a_{k,i}\|_2\|\mathbf a_{k,j}\|_2.
\end{equation}
Combining the two inequalities yields
\begin{equation}
\label{eq:perturbed_coherence_bound}
\widetilde{\mu}_{ij}
\ge
\max\left\{
0,\,
\frac{
\mu_{ij}
-
2\epsilon_{a,k}
-
\epsilon_{a,k}^{2}
}{
(1+\epsilon_{a,k})^{2}
}
\right\}.
\end{equation}

Equation~\eqref{eq:perturbed_coherence_bound} shows that moderate hardware errors degrade the guaranteed lower bound on air-region column coherence but do not alter the physical origin of background-column redundancy.
Therefore, the air-region columns can still preserve high coherence when the aggregate hardware-induced perturbation remains moderate.

To further assess this effect, Table~\ref{tab:hardware_coherence_test} reports a lightweight perturbation test on the propagation operators.
The hardware-error level $\epsilon_{\rm hw}$ is a normalized aggregate mismatch parameter satisfying $|\Delta\mathbf H_{\ell,k}|_F/|\mathbf H_{\ell,k}|_F\le\epsilon_{\rm hw}$, which jointly models gain/phase distortion, coupling-like array mixing, and additive residual propagation errors. The cases $\epsilon_{\rm hw}=0.03$ and $0.05$ represent moderate hardware perturbations, while $\epsilon_{\rm hw}=0.08$ and $0.10$ are included as stress tests.

\begin{table}[t]
\centering
\caption{Effect of hardware perturbations on air-region column coherence.}
\label{tab:hardware_coherence_test}
\begin{tabular}{c|c|c|c}
\hline
\(\epsilon_{\rm hw}\)
&
\(\overline{\widetilde{\mu}}_{\rm air}\)
&
\({\rm LB}\)
&
\(\min(\widetilde{\mu}_{\rm air}-{\rm LB})\)
\\
\hline
0.03 & 0.995 & 0.604 & 0.327 \\
0.05 & 0.991 & 0.406 & 0.511 \\
0.08 & 0.985 & 0.171 & 0.711 \\
0.10 & 0.978 & 0.042 & 0.824 \\
\hline
\end{tabular}
\vspace{3mm}
\end{table}

Here, $\overline{\widetilde{\mu}}_{\rm air}$ denotes the Monte Carlo average of the empirical air-region column coherence, ${\rm LB}$ denotes the theoretical lower bound in~\eqref{eq:perturbed_coherence_bound}, and $\min(\widetilde{\mu}_{\rm air}-{\rm LB})$ denotes the minimum trial-wise empirical margin above the bound.
As shown in Table~\ref{tab:hardware_coherence_test}, the air-region column coherence
remains high even under stress-test perturbations, and the positive empirical margins verify that the observed coherence is consistent with the theoretical lower bound.
\end{remark}

\vspace{-1em}

\subsection{ASR Columns: Weak Coherence}

\begin{figure}[t]
  \centering
  \includegraphics[width=0.98\columnwidth]{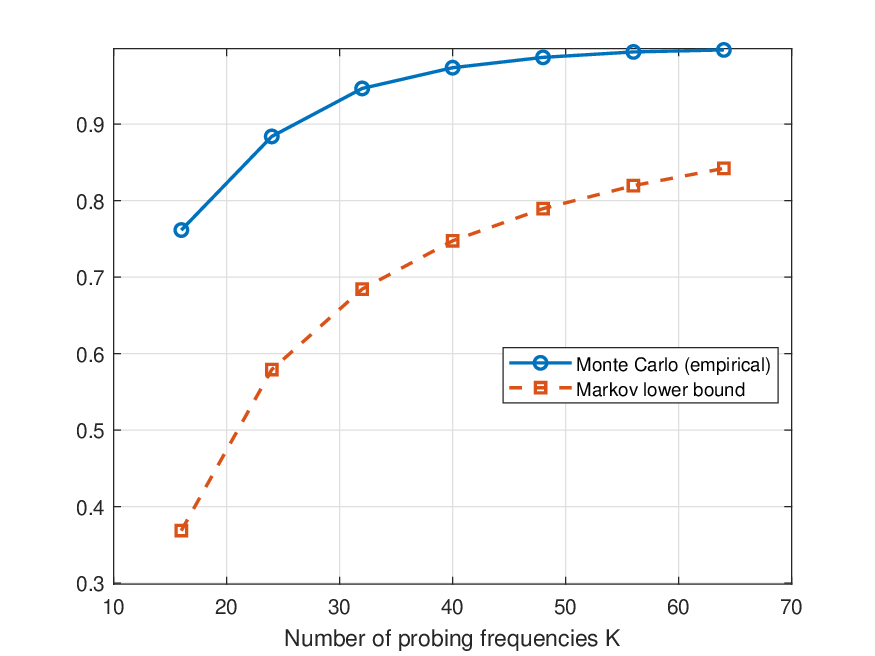}
  \caption{Probability of low residual coherence versus $K$.}
  \label{fig:markov}
  \vspace{5mm}
\end{figure}

\begin{lemma}\label{lem:ASR_indep}
Columns corresponding to the actual scatterer region are weakly correlated, and their coherence is directly affected by the number of probing frequencies.
\end{lemma}

\begin{proof}
Recall $\mathbf a_{k,j}=\mathbf u_j\circ \mathbf v_j$. By (A.1),
\begin{equation}\label{eq:NCC_factor_ASR_slim}
\mu_{ij}
=
\frac{|\mathbf u_i^\mathrm{H}\mathbf u_j|}{\|\mathbf u_i\|_2\|\mathbf u_j\|_2}\,
\frac{|\mathbf v_i^\mathrm{H}\mathbf v_j|}{\|\mathbf v_i\|_2\|\mathbf v_j\|_2}.
\end{equation}

Within the spatial-coherence regime (Appendix~B), there exists $\epsilon_v\in(0,1)$, which quantifies the maximal deviation of the
normalized correlation of the propagation factor $\mathbf v$ from unity, such that
\begin{equation}\label{eq:v_corr_bound}
1-\epsilon_v \le
\frac{|\mathbf v_i^\mathrm{H}\mathbf v_j|}{\|\mathbf v_i\|_2\|\mathbf v_j\|_2}
\le 1 .
\end{equation}

Thus, in the locally coherent propagation regime, the column coherence is dominated by the in-domain factor 
$u_j$, whose decorrelation is induced by multiple scattering and multi-frequency phase mixing.

For a homogeneous scatterer, the contrast is spatially constant within the \ac{ASR}.
Let $\mathbf W \triangleq \mathbf X_k^{\mathrm T}\mathbf H_{1,k}^{\mathrm T}\in\mathbb C^{T\times N}$ and
$\mathbf Z_{\chi,k} \triangleq \mathbf G_k^{D}\,\mathrm{diag}(\boldsymbol\chi)\in\mathbb C^{N\times N}$.
Then
\begin{equation}
\mathbf{e} \triangleq \mathbf W(\mathbf I-\mathbf Z_{\chi,k})^{-{\mathrm T}}\in\mathbb C^{T\times N}.
\label{eq:ASR_inverse}
\end{equation}

Since the scattered-field energy and the number of effective scattering interactions are bounded,
spectral radius $\rho(\mathbf Z_{\chi,k})<1$ holds and the Neumann series converges:
\begin{equation}\label{eq:Neumann_series_u}
\mathbf e
= \mathbf W\sum_{n=0}^{\infty} \big(\mathbf Z_{\chi,k}^{\mathrm T}\big)^{n}.
\end{equation}

The zeroth-order term $\mathbf W$ corresponds to propagation without contrast-induced coupling,
while the higher-order terms represent in-domain multiple scattering.
The additional decorrelation inside the \ac{ASR} is therefore governed by the in-domain Green's function
$\mathbf G_k^D$.

Let $\mathbf G_{k,\mathrm{ASR}}^{D}\in\mathbb C^{N_d\times N_d}$ denote the in-domain Green's submatrix
restricted to the ASR grid with $N_d$ pixels, and let $\mathbf g_{k,i}\triangleq [\mathbf G_{k,\mathrm{ASR}}^{D}]_{[:,i]}\in\mathbb C^{N_d \times 1}$.
Using the large-argument expansion of the Hankel function, the inner product between two distinct columns can be written as:
\begin{equation}\label{eq:G_inner_random_phase}
\mathbf g_{k,i}^{\mathrm H}\mathbf g_{k,j}
= \sum_{m=1}^{N_d} \alpha_{m,k}\,
\exp\!\big(-\mathrm j\,\phi_{ij,k}(\mathbf r_m)\big)
\;+\; \Delta_{k,ij},
\end{equation}
where
\begin{equation}
\phi_{ij,k}(\mathbf r_m)\triangleq k_b\Big(\|\mathbf r_m-\mathbf r_j\|_2-\|\mathbf r_m-\mathbf r_i\|_2\Big),
\end{equation}
and $k_b$ denotes the background wavenumber at tone $k$.
Define $d_{m,i}\triangleq \|\mathbf r_m-\mathbf r_i\|_2$ and $d_{m,j}\triangleq \|\mathbf r_m-\mathbf r_j\|_2$.
The amplitude weight induced by the large-argument Hankel expansion is
\begin{equation}\label{eq:alpha_def}
\alpha_{m,k}\triangleq \frac{2|b_k|^2}{\pi k_b\sqrt{d_{m,i}d_{m,j}}},
\qquad b_k\triangleq -\frac{\mathrm j}{4}k_b^2.
\end{equation}

This approximation is valid when $k_b d \gg 1$, a condition that typically holds in the $28\,\mathrm{GHz}$ band.

The remainder term $\Delta_{k,ij}$ collects the higher-order terms ignored in the Hankel asymptotic:
\begin{equation}\label{eq:Delta_def}
\Delta_{k,ij}\triangleq
\sum_{m=1}^{N_d}\Big(
G_k(\mathbf r_m,\mathbf r_i)^{*}\,G_k(\mathbf r_m,\mathbf r_j)
-\alpha_{m,k}\,e^{-\mathrm j\,\phi_{ij,k}(\mathbf r_m)}
\Big).
\end{equation}

Define the geometry-induced effective phase
\begin{equation}\label{eq:phi_eff}
\varphi_{ij,k}\triangleq
\arg\!\left(\sum_{m=1}^{N_d}\alpha_{m,k}\,e^{-\mathrm j\,\phi_{ij,k}(\mathbf r_m)}\right).
\end{equation}
Across frequencies, we characterize the multi-frequency phase-mixing effect under the standard frequency-incoherent regime,
where the tone spacing $\Delta f$ exceeds the coherence bandwidth so that cross-tone correlations are weak.
We model the aggregate residual by an i.i.d.\ random phase $\theta_k\sim\mathrm{Unif}[0,2\pi)$ across tones,
independent of the geometry, and define the cross-tone phasor as:
\begin{equation}\label{eq:cross_tone_phasor}
z_k(i,j)\triangleq \exp\!\big(\mathrm j(\varphi_{ij,k}+\theta_k)\big),\qquad
\bar z(i,j)\triangleq \frac{1}{K}\sum_{k=1}^{K} z_k(i,j).
\end{equation}

In practical \ac{ISAC} deployments, the probing tones are not necessarily strictly uncorrelated. 
We therefore consider a more general second-order correlation model:
\begin{equation}
\label{eq:freq_corr_assump}
\begin{aligned}
|z_k(i,j)| &= 1,\\
\mathbb E\!\left[z_k(i,j)\right] &= 0,\\
\mathbb E\!\left[z_k(i,j)z_\ell(i,j)^{*}\right]
&=
\rho_{k\ell},
\qquad
|\rho_{k\ell}|\le \rho_f,\quad k\neq \ell,
\end{aligned}
\end{equation}
where $\rho_f\in[0,1]$ characterizes the residual cross-tone correlation. 
Since $\rho_{k\ell}$ may be complex-valued, we bound the cross-tone contribution through its real part, with $\Re\{\rho_{k\ell}\}\le |\rho_{k\ell}|\le\rho_f$. 
Under \eqref{eq:freq_corr_assump},
\begin{equation}
\label{eq:zbar_second_moment_corr}
\begin{aligned}
\mathbb E\!\left[|\bar z(i,j)|^{2}\right]
&=
\frac{1}{K^2}
\left(
K+
\operatorname{Re}\left\{
\sum_{k\neq \ell}\rho_{k\ell}
\right\}
\right)
\\
&\le
\frac{1}{K^2}
\left(
K+
\sum_{k\neq \ell}|\rho_{k\ell}|
\right)
\\
&\le
\frac{1}{K}
+
\frac{K-1}{K}\rho_f.
\end{aligned}
\end{equation}
This bound shows that the decorrelation gain is governed not only by the number of probing tones $K$ but also by the residual frequency correlation $\rho_f$.
Equivalently, we define the effective number of independent probing tones as
\begin{equation}
\label{eq:keff}
K_{\mathrm{eff}}
\triangleq
\frac{K}{1+(K-1)\rho_f},
\end{equation}
so that
\begin{equation}
\label{eq:zbar_keff}
\mathbb E\!\left[|\bar z(i,j)|^{2}\right]
\le
\frac{1}{K_{\mathrm{eff}}}.
\end{equation}

The ideal frequency-incoherent case used in the following numerical verification corresponds to $\rho_f=0$.
Then $K_{\mathrm{eff}}=K$, and \eqref{eq:zbar_keff} reduces to
\begin{equation}
\label{eq:zbar_second_moment}
\mathbb E\!\left[|\bar z(i,j)|^{2}\right]
\le
\frac{1}{K}.
\end{equation}
Thus, the $1/K$ decay should be interpreted as the most favorable case with effectively independent probing tones.
When the tones are partially correlated, the same argument still holds with $K$ replaced by
$K_{\mathrm{eff}}$. 
Hence, \ac{ASR} columns do not immediately lose their informativeness under frequency correlation; instead, the multi-frequency decorrelation gain gradually weakens as $\rho$ increases.
In the extreme case $\rho_f\rightarrow 1$, $K_{\mathrm{eff}}\rightarrow 1$, and the
multi-frequency averaging gain almost vanishes.

To illustrate the favorable frequency-incoherent case, Fig.~\ref{fig:markov} presents a numerical verification based on the phase-randomization model.
The horizontal axis is the number of probing frequencies, and the vertical axis is the probability that the residual coherence between two \ac{ASR}-related columns remains below the fixed threshold $0.25$.
The blue curve shows the Monte Carlo empirical probability, whereas the orange curve shows a conservative analytical lower bound obtained from the second-moment analysis via Markov's inequality.
As $K=K_{\mathrm{eff}}$ increases in this ideal case, the residual coherence is suppressed and the inter-column correlation approaches the weakly correlated regime.
\end{proof}

Consequently, these observations imply that once the frequency-dimension diversity is fixed by system constraints, the spatial selection of the \ac{ROI} becomes the predominant factor in governing the problem's ill-posedness. In this regime, the system's conditioning is highly dictated by the inclusion of strongly coherent background columns, underscoring \ac{ROI} restriction as the decisive mechanism for stabilizing the material reconstruction process.

\begin{figure}[t]
  \centering
  \includegraphics[width=0.8\columnwidth]{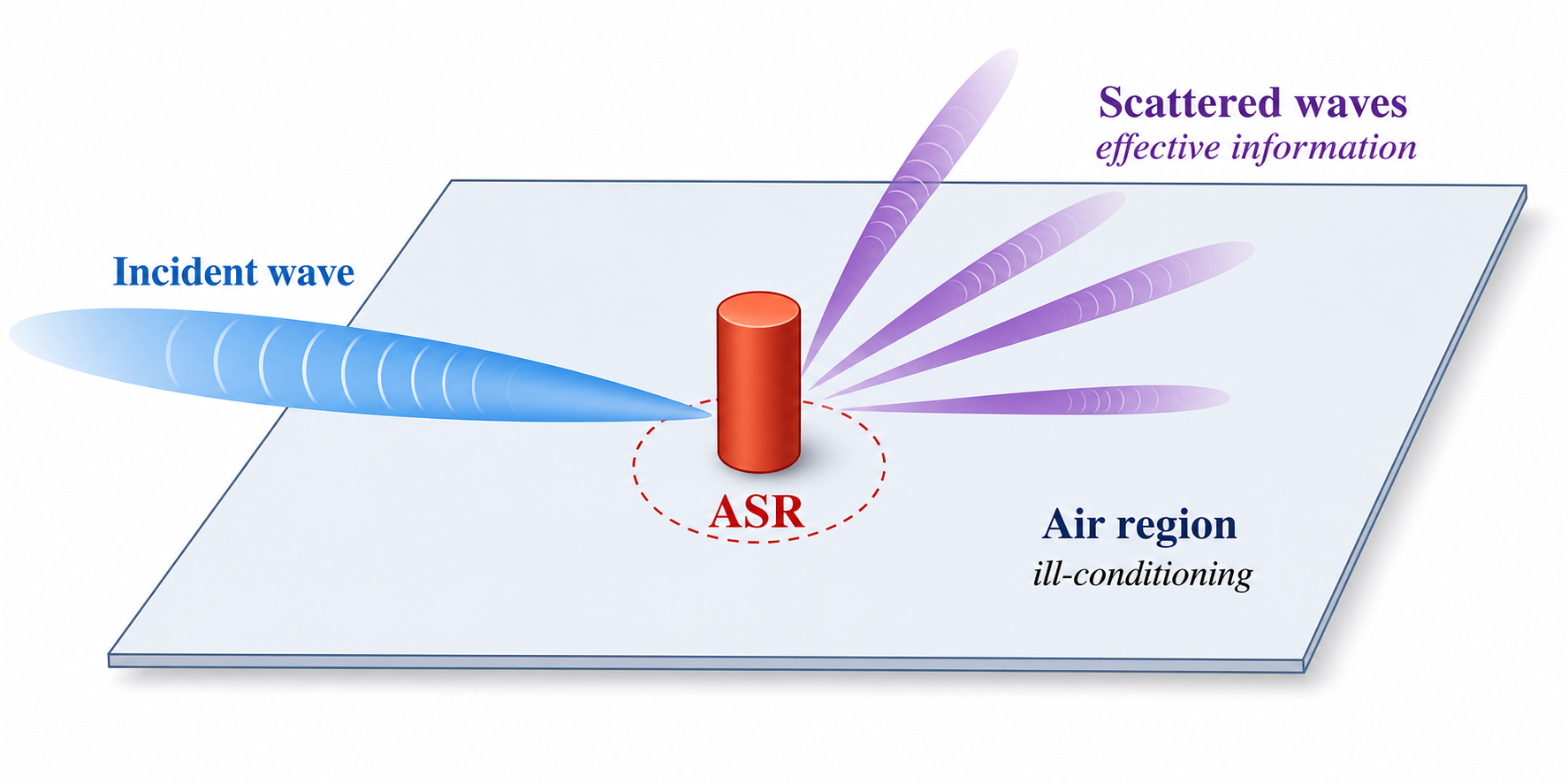}
  \caption{Illustration of the distinction of the ASR and air-background.}
  \vspace{3mm}
  
\end{figure}

\vspace{-1em}

\subsection{Hilbert-Space Interpretation of the Coherence Mechanism}

\subsubsection{Operator-Induced Coherence in Background and ASR Subspaces}

We now reinterpret Lemmas~1 and~2 from a Hilbert-space operator viewpoint.
The purpose is not to introduce a new discretization-dependent coherence metric, but to explain why the coherence contrast between background columns and \ac{ASR} columns follows from the underlying propagation and scattering operators.

Let \(\mathcal X_D=L^2(D)\) denote the contrast Hilbert space under the \ac{2D} TM\(_z\) model, and let \(\mathcal Y_k\) denote the \ac{CSI} observation Hilbert space.
For the multiplication operators below, we assume $\chi\in L^\infty(D)\cap L^2(D)$ and $f\in L^\infty(D)\cap L^2(D)$, with localized perturbations extended by zero when necessary, which ensures that $\mathcal M_\chi$ and $\mathcal M_f$ are bounded on $L^2(D)$, $L^\infty(D)$ garuntees the boundedness.

The continuous scattering operator corresponding to the discretized scattering matrices
\(\mathbf A_k^{(n)}\) is denoted by
\begin{equation}
    \mathcal A^{(n)}_k:\mathcal X_D\rightarrow \mathcal Y_k ,
    \label{eq:continuous_scattering_operator}
\end{equation}
which maps contrast functions to \ac{CSI} responses.

For any subspace \(\mathcal U\subset \mathcal X_D\), we define the operator-induced coherence as
\begin{equation}
    \mu_{\mathcal A_k}(\mathcal U)
    =
    \sup_{\substack{
    f,g\in \mathcal U\\
    f\perp g\\
    \mathcal A_k f\neq 0,\ \mathcal A_k g\neq 0
    }}
    \frac{
    \left|
    \left\langle
    \mathcal A_k f,
    \mathcal A_k g
    \right\rangle_{\mathcal Y_k}
    \right|
    }{
    \|\mathcal A_k f\|_{\mathcal Y_k}
    \|\mathcal A_k g\|_{\mathcal Y_k}
    },
    \label{eq:operator_induced_coherence}
\end{equation}
where \(f\) and \(g\) are two contrast perturbation functions in \(\mathcal U\), and \(f\perp g\) means that they are orthogonal in the contrast Hilbert space \(\mathcal X_D\).
When \(f\) and \(g\) are localized within two discretization cells of \(D\), \eqref{eq:operator_induced_coherence} reduces to the normalized column correlation in~\eqref{eq:ncc_def}.
Thus, the discrete column-coherence behavior analyzed in Lemmas~1 and~2 is a finite-dimensional manifestation of this operator-level quantity.

In the air region, material-dependent in-domain scattering is absent.
Hence, for a given tone \(k\) and pilot \(x_{k,t}\), the linearized background map reduces to the propagation-only operator
\begin{equation}
\mathcal A^{\rm air}_{k,t} f
=
\mathcal H_{2,k}
\left(
f\mathcal H_{1,k}x_{k,t}
\right),
\label{eq:air_operator_propagation_form}
\end{equation}
where \(\mathcal H_{1,k}\) and \(\mathcal H_{2,k}\) denote the Tx-to-domain and domain-to-Rx Green propagation operators, respectively.
Since the antenna arrays are separated from the sensing domain, the external Green
kernels are smooth and nonsingular over \(D\). Thus, \(\mathcal H_{2,k}\) acts as a compact smoothing operator, and \(\mathcal A^{\rm air}_{k,t}\) is compact as its composition with the bounded multiplication operator induced by \(\mathcal H_{1,k}x_{k,t}\).
Consequently,
\begin{equation}
\sigma_m\left(\mathcal A^{\rm air}_{k,t}\right)
\rightarrow 0,
\qquad
m\rightarrow \infty .
\label{eq:compact_air_operator_singular_decay}
\end{equation}
Therefore, the propagation-only background map compresses many object-space directions into a low-effective-dimensional response set, which yields near-null directions and strong redundancy among air-region columns after discretization.

The local high coherence of neighboring background columns is further induced by the smoothness of the Green kernels. 
If \(\psi_{k,t}(\mathbf r)\in\mathcal Y_k\) denotes the \ac{CSI} response generated by a unit localized background component at \(\mathbf r\), then for nearby
\(\mathbf r_i,\mathbf r_j\in D\),
\begin{equation}
    \left|
    \left\langle
    \frac{\psi_{k,t}(\mathbf r_i)}
    {\|\psi_{k,t}(\mathbf r_i)\|_{\mathcal Y_k}},
    \frac{\psi_{k,t}(\mathbf r_j)}
    {\|\psi_{k,t}(\mathbf r_j)\|_{\mathcal Y_k}}
    \right\rangle_{\mathcal Y_k}
    \right|
    \rightarrow 1,
    \qquad
    \mathbf r_i\rightarrow \mathbf r_j .
    \label{eq:local_background_coherence}
\end{equation}
Thus, compactness explains the weak distinguishability of the background subspace, while the smooth Green kernel explains the local high coherence of neighboring air-region columns.

In contrast, inside the actual scatterer region, the forward map is no longer propagation-only.
Let \(D_s\subset D\) denote the \ac{ASR}, and let \(\mathcal X_s=L^2(D_s)\) be the corresponding contrast subspace.
For a fixed contrast distribution \(\chi\), or within one Born-iterative linearization,
the corresponding \ac{ASR} sensing map can be written as
\begin{equation}
\begin{aligned}
\mathcal A^{\rm ASR}_{\chi,k,t} f
&=
\mathcal H_{2,k}
\mathcal M_f
\mathscr S_{\chi,k}
\mathcal H_{1,k}x_{k,t}, \\
\mathscr S_{\chi,k}
&=
\left(\mathcal I-\mathcal G_k\mathcal M_{\chi}\right)^{-1},
\end{aligned}
\label{eq:asr_operator_resolvent_form}
\end{equation}
where \(f\in L^2(D_s)\) denotes a localized contrast perturbation, \(\mathcal G_k\) is the in-domain Green operator, \(\mathcal M_f\) and \(\mathcal M_{\chi}\) are multiplication operators, and \(\mathscr S_{\chi,k}\) is the Lippmann--Schwinger resolvent evaluated at the current contrast \(\chi\). 
Compared with the air-region operator, \(\mathcal A^{\rm ASR}_{\chi,k,t}\) contains the internal resolvent \(\mathscr S_{\chi,k}\). 
Thus, \ac{ASR} responses are shaped by material support, in-domain Green-function coupling, and resolvent-induced modal reweighting.

To make this spectral effect explicit, define the internal scattering operator as \(\mathcal L_{\chi,k}\triangleq\mathcal G_k\mathcal M_{\chi}\).
If \(q_\ell\) is an internal scattering mode satisfying \(\mathcal L_{\chi,k}q_\ell=\Lambda_\ell q_\ell\), then, in the non-resonant regime where \(\mathcal I-\mathcal L_{\chi,k}\) is invertible, the resolvent acts on this mode as
\begin{equation}
\label{eq:resolvent_modal_reweighting}
\mathscr S_{\chi,k}q_\ell
=
\frac{1}{1-\Lambda_\ell}q_\ell .
\end{equation}
Thus, each internal scattering mode is selectively amplified or attenuated by
the spectral factor \(1/(1-\Lambda_\ell)\).
Modes with eigenvalues closer to one are more strongly reweighted, whereas modes with small \(|\Lambda_\ell|\) remain closer to the propagation-dominated response.
Consequently, different localized \ac{ASR} perturbations \(f\) can excite different observable mixtures of the internal modes after the action of \(\mathcal H_{1,k}\), \(\mathscr S_{\chi,k}\), \(\mathcal M_f\), and \(\mathcal H_{2,k}\).
This material-induced spectral diversity makes \ac{ASR}-related columns comparatively more distinguishable after discretization, providing an operator-level explanation for Lemma~2.

\subsubsection{Extension to 3D Multi-Polarization Scattering}

From the Hilbert-space viewpoint, the above coherence mechanism extends to the \ac{3D} vector \ac{EM} setting.
In the two-dimensional TM\(_z\) model, the field is scalar and the integral equation is defined over \(D\subset\mathbb R^2\).
In the \ac{3D} homogeneous and isotropic case, the unknown contrast remains scalar-valued, \(\chi(\mathbf r)\in L^2(D)\) with \(D\subset\mathbb R^3\), whereas the electric field belongs to \(L^2(D;\mathbb C^3)\).

The \ac{3D} frequency-domain Lippmann--Schwinger equation is
\begin{equation}
    \mathbf E_k^t(\mathbf r)
    =
    \mathbf E_k^i(\mathbf r)
    +
    k_b^2
    \int_D
    \overline{\overline{\mathbf G}}_k(\mathbf r,\mathbf r')
    \chi(\mathbf r')
    \mathbf E_k^t(\mathbf r')
    d\mathbf r',
    \label{eq:ls_equation_3d}
\end{equation}
where \(\mathbf E_k^t,\mathbf E_k^i\in L^2(D;\mathbb C^3)\) denote the total and incident electric fields, respectively, and \(d\mathbf r'\) is the \ac{3D} volume element.
The dyadic Green function is
\begin{equation}
    \overline{\overline{\mathbf G}}_k(\mathbf r,\mathbf r')
    =
    \left(
    \mathbf I_3
    +
    \frac{1}{k_b^2}\nabla\nabla^{\mathrm T}
    \right)
    \frac{
    e^{-\mathrm j k_b\|\mathbf r-\mathbf r'\|}
    }{
    4\pi\|\mathbf r-\mathbf r'\|
    } .
    \label{eq:dyadic_green_tensor_3d}
\end{equation}
Compared with the \ac{2D} scalar case, the dyadic Green function additionally captures coupling among electric-field components, and the observation model must include polarization projection.

In the air region, no material-dependent in-domain scattering is present.
The corresponding propagation-only map can be written as
\begin{equation}
    \mathcal A_{k,t}^{{\rm air},3D} f
    =
    \mathcal P_{k,{\rm pol}}
    \mathcal H_{2,k}^{3D}
    \left(
    f\mathbf E_{k,t}^{i}
    \right),
    \label{eq:air_operator_3d}
\end{equation}
where \(\mathcal H_{2,k}^{3D}\) is the domain-to-Rx propagation operator and \(\mathcal P_{k,{\rm pol}}\) is the multi-polarization observation projection. 
For example, suppose that the \(m\)-th receive antenna supports the polarization vector \(\mathbf p_{m,\ell}\in\mathbb C^3\), \(\ell=1,\ldots,N_{\rm pol}\), where \(N_{\rm pol}\) denotes the number of polarization channels supported by each receive antenna. Then
\begin{equation}
    [\mathcal P_{k,{\rm pol}}\mathbf u]_{m,\ell}
    =
    \mathbf p_{m,\ell}^{\mathrm H}
    \mathbf u(\mathbf r_m).
    \label{eq:polarization_projection}
\end{equation}
Since \(\mathcal P_{k,{\rm pol}}\) is a finite-dimensional projection, it is bounded.

Because the Tx/Rx arrays are separated from the sensing domain, the external propagation kernels associated with \(\mathcal H_{1,k}^{3D}\) and \(\mathcal H_{2,k}^{3D}\) are smooth and nonsingular over \(D\). 
Hence, \(\mathcal H_{2,k}^{3D}\) remains a compact smoothing propagation operator.
This compactness statement concerns only the external Tx/Rx propagation operators; the in-domain dyadic Green operator in the Lippmann--Schwinger equation still contains its standard singularity and is understood in the usual volume-integral-equation sense.
Since the multiplication map \(f\mapsto f\mathbf E_{k,t}^{i}\) and the projection \(\mathcal P_{k,{\rm pol}}\) are bounded, \(\mathcal A_{k,t}^{{\rm air},3D}\) is compact as a composition of bounded operators with a compact propagation operator.
Therefore, the compact smoothing mechanism responsible for air-region coherence in the \ac{2D} TM\(_z\) case is preserved in the \ac{3D} multi-polarization setting.

For the \ac{ASR}, the contrast multiplication operator acts on vector electric fields as
\begin{equation}
    \mathcal M_{\chi}^{3D}\mathbf u
    =
    \chi(\mathbf r)\mathbf u(\mathbf r),
    \qquad
    \mathbf u\in L^2(D;\mathbb C^3).
    \label{eq:contrast_multiplication_3d_asr}
\end{equation}
Let \(\mathcal G^{3D}\) denote the in-domain dyadic Green operator. 
The \ac{3D} Lippmann--Schwinger resolvent is then
\begin{equation}
    \mathscr S_{k,\chi}^{3D}
    =
    \left(
    I-\mathcal G^{3D}\mathcal M_{\chi}^{3D}
    \right)^{-1}.
    \label{eq:resolvent_3d_asr}
\end{equation}
Accordingly, the \ac{ASR}-restricted sensing map is
\begin{equation}
    \mathcal A_{\chi,k,t}^{{\rm ASR},3D} f
    =
    \mathcal P_{k,{\rm pol}}
    \mathcal H_{2,k}^{3D}
    \mathcal M_f^{3D}
    \mathscr S_{k,\chi}^{3D}
    \mathcal H_{1,k}^{3D} \mathbf x_{k,t},
    \label{eq:asr_sensing_operator_3d}
\end{equation}
where \(f\in L^2(D_s)\) denotes a localized contrast perturbation. 
This map has the same resolvent-mediated structure as in the 2D TM\(_z\) case, while
the dyadic Green operator and the polarization projection modify the vector-field coupling and observation dimension.
Therefore, the \ac{3D} formulation preserves the same operator-theoretic distinction between compact propagation-dominated air responses and material-supported \ac{ASR}
responses.
Consequently, the structural conclusions of Lemmas~1 and~2 remain valid at the operator-structure level in the \ac{3D} multi-polarization setting.

\vspace{-0.5em}
\section{Empirical Validation via ROI-Constrained Subproblems}

\subsection{ROI-Constrained Algorithm}

Based on the coherence analysis in Section~III, we use an \ac{LSM}-guided \ac{ROI} to suppress highly coherent background columns and solve the quantitative update only in the resulting reduced subspace.
The resulting \ac{ROI}-\ac{QP} procedure is used as an analysis-guided implementation to
validate the proposed subspace principle, rather than as the main algorithmic contribution.

We first compute the multi-frequency \ac{LSM} indicator \(\mathcal J_K\) from the multi-static response matrix \(\mathbf U_k\) introduced in Section~II-B.
The indicator is converted into the normalized score
\begin{equation}
\label{eq:roi_score}
s_p
=
\frac{[\mathcal{J}_K]_p - \mathcal{J}_{K,\max}}
{\mathcal{J}_{K,\max}-\mathcal{J}_{K,\min}+\epsilon},
\qquad p=1,\ldots,N,
\end{equation}
where a smaller \(s_p\) indicates a more scatterer-like pixel.
The threshold \(\eta\) is selected by applying a trimmed maximum-gap rule to \(\{s_p\}_{p=1}^{N}\), and the \ac{ROI} index set is defined as
\begin{equation}
\label{eq:roi_index}
\mathcal I
\triangleq
\{p:s_p\le \eta\},
\qquad
P\triangleq |\mathcal I|.
\end{equation}
Here, \(\epsilon\) avoids numerical degeneracy in the normalization, while the trimming step prevents extreme outliers from dominating the maximum-gap threshold.

\begin{algorithm}[!t]
\caption{LSM-Guided ROI-QP Reconstruction}
\label{alg:roi_qp_method}
\LinesNumbered
\DontPrintSemicolon

\KwIn{Measurements $\{\mathbf y_k\}_{k=1}^{K}$, LSM response matrices
$\{\mathbf U_k\}_{k=1}^{K}$, sensing operators $\{\mathbf A_k\}_{k=1}^{K}$,
normalization constant $\epsilon_{\mathrm{norm}}$, trimming ratio
$q_{\mathrm{trim}}$, regularization parameters $(\alpha,\beta)$, stopping
parameters $(\tau_{\mathrm{err}},M)$.}

\KwOut{Reconstructed contrast vector $\hat{\boldsymbol\chi}$.}

Compute the multi-frequency LSM indicator $\mathcal{J}_K$ and determine the
ROI index set $\mathcal I$ using \eqref{eq:roi_score}--\eqref{eq:roi_index}
with $\epsilon_{\mathrm{norm}}$ and $q_{\mathrm{trim}}$.\;

Stack the multi-frequency observations:
\[
\mathbf Y
\gets
\begin{bmatrix}
\mathbf y_1\\
\vdots\\
\mathbf y_K
\end{bmatrix},
\qquad
\widetilde{\mathbf Y}
\gets
\begin{bmatrix}
\Re\{\mathbf Y\}\\
\Im\{\mathbf Y\}
\end{bmatrix}.
\]

Initialize $n\gets1$, $\boldsymbol\chi^{(0)}\gets\mathbf 0$,
$\boldsymbol\chi^{\mathrm{sub},(0)}
\gets[\boldsymbol\chi^{(0)}]_{\mathcal I}=\mathbf 0$,
$e_{\mathrm{stop}}\gets+\infty$, and
$\hat{\boldsymbol\chi}\gets\boldsymbol\chi^{(0)}$.\;

\While{$e_{\mathrm{stop}}>\tau_{\mathrm{err}}$ \textbf{\rm and} $n\le M$}{

Update the Born-iterative sensing matrices
$\{\mathbf A_k^{(n)}\}_{k=1}^{K}$ from \eqref{eq:Ak} using
$\boldsymbol\chi^{(n-1)}$; the initialization
$\boldsymbol\chi^{(0)}=\mathbf 0$ corresponds to the first-Born
initialization.\;

Construct the ROI-restricted sensing matrix:
\[
\mathbf A^{\mathrm{sub},(n)}
\gets
\begin{bmatrix}
[\mathbf A_1^{(n)}]_{[:,\mathcal I]}\\
\vdots\\
[\mathbf A_K^{(n)}]_{[:,\mathcal I]}
\end{bmatrix}.
\]

Convert it into the real-valued form:
\[
\widetilde{\mathbf A}^{\mathrm{sub},(n)}
\gets
\begin{bmatrix}
\Re\{\mathbf A^{\mathrm{sub},(n)}\}\\
\Im\{\mathbf A^{\mathrm{sub},(n)}\}
\end{bmatrix}.
\]

Construct the ROI graph Laplacian $\mathbf L_R$.\;

Solve the reduced real-valued ROI-QP in \eqref{eq:opt_real} using
$\widetilde{\mathbf A}^{\mathrm{sub},(n)}$, $\widetilde{\mathbf Y}$,
and $\mathbf L_R$, and obtain $\boldsymbol\chi^{\mathrm{sub},(n)}$.\;

Extend the ROI estimate to the full domain:
\[
\boldsymbol\chi^{(n)}[\mathcal I]
\gets
\boldsymbol\chi^{\mathrm{sub},(n)},
\qquad
\boldsymbol\chi^{(n)}[\mathcal I^c]
\gets
0.
\]

Set $\hat{\boldsymbol\chi}\gets\boldsymbol\chi^{(n)}$ and update
\[
e_{\mathrm{stop}}
\gets
\|\boldsymbol\chi^{\mathrm{sub},(n)}
-
\boldsymbol\chi^{\mathrm{sub},(n-1)}\|_2,
\qquad
n\gets n+1.
\]
}

\Return the latest estimate $\hat{\boldsymbol\chi}$.\;
\end{algorithm}

Using the \ac{ROI} indices, we form the reduced sensing matrix at each frequency by selecting the corresponding columns:
\begin{equation}
\label{eq:submatrix}
\mathbf A_k^{\mathrm{sub}}
\triangleq
\left[\mathbf A_k\right]_{[:,\mathcal I]}
\in
\mathbb C^{TN_r\times P}.
\end{equation}
The multi-frequency observations and reduced sensing matrices are then stacked as
\begin{equation}
\label{eq:full}
\mathbf Y
\triangleq
\begin{bmatrix}
\mathbf y_1\\
\vdots\\
\mathbf y_K
\end{bmatrix},
\qquad
\mathbf A^{\mathrm{sub}}
\triangleq
\begin{bmatrix}
\mathbf A_1^{\mathrm{sub}}\\
\vdots\\
\mathbf A_K^{\mathrm{sub}}
\end{bmatrix},
\end{equation}
where \(\mathbf Y\in\mathbb C^{KTN_r\times 1}\) and \(\mathbf A^{\mathrm{sub}}\in\mathbb C^{KTN_r\times P}\).

For simplicity, this work focuses on relative-permittivity reconstruction; therefore, the unknown contrast variables are treated as real-valued vectors.
Let \(\boldsymbol\chi^{\mathrm{sub}}\triangleq\boldsymbol\chi_{\mathcal I}\in \mathbb R^P\) denote the contrast vector restricted to the \ac{ROI}.
The original full-domain inverse problem is then replaced by the following real-valued \ac{ROI}-restricted model:
\begin{equation}
\label{eq:subproblem}
\widetilde{\mathbf Y}
=
\widetilde{\mathbf A}^{\mathrm{sub}}\boldsymbol\chi^{\mathrm{sub}}
+
\mathbf d,
\end{equation}
where \(\widetilde{\mathbf Y}\) and \(\widetilde{\mathbf A}^{\mathrm{sub}}\) denote the real-valued representations of the stacked observation vector \(\mathbf Y\) and the
\ac{ROI}-restricted sensing matrix \(\mathbf A^{\mathrm{sub}}\), respectively.
The vector \(\mathbf d\in\mathbb R^{2KTN_r\times 1}\) denotes the effective residual accounting for observation noise, \ac{ROI} truncation error, and model mismatch.

We solve the following regularized \ac{QP}:
\begin{equation}
\label{eq:opt_real}
\begin{aligned}
\min_{\boldsymbol\chi^{\mathrm{sub}}\in\mathbb R^P,\,\mathbf d\in\mathbb R^{2KTN_r}}
\quad
&
\frac{1}{2}\|\mathbf d\|_2^2
+
\frac{\alpha}{2}\|\boldsymbol\chi^{\mathrm{sub}}\|_2^2
+
\frac{\beta}{2}
(\boldsymbol\chi^{\mathrm{sub}})^{\mathrm T}
\mathbf L_R
\boldsymbol\chi^{\mathrm{sub}}
\\
\mathrm{s.t.}\quad
&
\widetilde{\mathbf A}^{\mathrm{sub}}
\boldsymbol\chi^{\mathrm{sub}}
-
\widetilde{\mathbf Y}
+
\mathbf d
=
\mathbf 0,
\end{aligned}
\end{equation}
where \(\alpha\ge0\) controls the Tikhonov regularization and \(\beta\ge0\) controls the graph-Laplacian smoothness over the \ac{ROI}. 
The \ac{ROI} graph Laplacian is defined as
\begin{equation}
\label{eq:roi_laplacian}
\mathbf L_R
=
\mathbf Q-\mathbf W_R,
\end{equation}
where \(\mathbf W_R\) denotes the adjacency matrix of the \ac{ROI} graph, and \(\mathbf Q=\operatorname{diag}(\mathbf W_R\mathbf 1)\) is the corresponding degree matrix~\cite{b44}.

The complex-valued observation model is converted into the real-valued form above by separating the real and imaginary parts.
This real-valued formulation is used in the \ac{ROI}-\ac{QP} reconstruction summarized in
Algorithm~\ref{alg:roi_qp_method}.

Using the \ac{ROI} prior, the reconstruction is performed only in a reduced subspace of size \(P=|\mathcal I|\), where \(P\ll N\).
The \ac{LSM}-based \ac{ROI} selection has complexity \(O\!\left(K(N_t^3+TNN_rN_t)\right)\).

After the \ac{ROI} is obtained, each outer iteration constructs the \ac{ROI}-restricted sensing operator and solves the reduced \ac{QP} with complexity \(O\!\left(K(P^3+TP^2+T(N_t+N_r)P)\right)\).
Therefore, for \(M\) outer iterations, the overall complexity is
\begin{eqnarray}
\!\!\!\!&\!\!\!\!&\!\!\!\!
O\Bigl(
K(N_t^3+TNN_rN_t) \nonumber \\
\!\!\!\!&\!\!\!\!&\!\!\!\!
\qquad\qquad+
MK(P^3+TP^2+T(N_t+N_r)P)
\Bigr).
\end{eqnarray}
Since \(P\ll N\), the \ac{ROI}-constrained update provides a substantial computational reduction compared with full-domain reconstruction.

\subsection{Conditioning Implication of ROI Restriction}

In the preceding subsection, the LSM indicator is used to construct a
coarse ROI index set \(\mathcal I\), and the subsequent ROI-QP update is
performed only over this reduced subspace. In other words, the original
full-domain inverse problem is replaced by an ROI-restricted subproblem
with \(P=|\mathcal I|\) unknowns. To explain why this subspace restriction
improves numerical stability, we briefly characterize its conditioning
effect from the viewpoint of effective column coherence.

The LSM-derived ROI may not exactly coincide with the true ASR; it may
exclude part of the scatterer support or include a small number of
background pixels. Therefore, we do not assume a perfectly matched ROI.
For a fixed frequency tone \(k\), let
\(\mathbf A^{\mathrm{sub}}=\mathbf A_k^{\mathrm{sub}}
\in\mathbb C^{TN_r\times P}\) denote the ROI-restricted sensing matrix
constructed in the preceding subsection. The multi-frequency case follows
by replacing \(\mathbf A^{\mathrm{sub}}\) with the stacked operator in
\eqref{eq:full}.

Let \(\Omega=\{1,\ldots,N\}\) denote the full pixel set and
\(S_{\mathrm{true}}\subset\Omega\) denote the true ASR index set, with
\(k_{\mathrm{true}}=|S_{\mathrm{true}}|\). For the LSM-selected ROI index
set \(\mathcal I\subset\Omega\), define
\begin{equation}
\mathcal T
\triangleq
\mathcal I\cap S_{\mathrm{true}},
\qquad
\tau
\triangleq
|\mathcal T|.
\end{equation}
Here, \(\mathcal T\) denotes the part of the ROI that truly belongs to
the ASR. The corresponding recall and precision are
\begin{equation}
\mathcal R
\triangleq
\frac{\tau}{k_{\mathrm{true}}},
\qquad
p_{\mathrm{prec}}
\triangleq
\frac{\tau}{P}.
\end{equation}
Thus, the ROI size can be written as
\begin{equation}
P
=
\frac{\mathcal R}{p_{\mathrm{prec}}}
k_{\mathrm{true}}.
\end{equation}
This relation shows that, for a fixed true ASR size, a lower precision
corresponds to more background pixels included in the ROI and hence a
larger reduced subproblem.

Next, normalize the ROI sub-operator as
\begin{equation}
\bar{\mathbf A}^{\mathrm{sub}}
=
\mathbf A^{\mathrm{sub}}\mathbf D^{-1},
\qquad
\mathbf D
=
\operatorname{diag}
\big(
\|\mathbf a_1\|_2,\ldots,\|\mathbf a_P\|_2
\big).
\end{equation}
The normalized Gram matrix is
\begin{equation}
\mathbf\Phi_{\mathcal I}
\triangleq
(\bar{\mathbf A}^{\mathrm{sub}})^{\mathrm H}
\bar{\mathbf A}^{\mathrm{sub}}.
\end{equation}
Since \(\bar{\mathbf A}^{\mathrm{sub}}\) is column-normalized, the diagonal
entries of \(\mathbf\Phi_{\mathcal I}\) are one. Define the effective
mutual coherence inside the ROI as
\begin{equation}
\mu_{\mathrm{eff}}
\triangleq
\max_{i\neq j}
\left|
[\mathbf\Phi_{\mathcal I}]_{ij}
\right|.
\end{equation}
Then the \(j\)-th Gershgorin radius of \(\mathbf\Phi_{\mathcal I}\)
satisfies
\begin{equation}
R_j
=
\sum_{i\neq j}
\left|
[\mathbf\Phi_{\mathcal I}]_{ji}
\right|
\le
(P-1)\mu_{\mathrm{eff}}.
\end{equation}
By the Gershgorin disk theorem, under the sufficient condition
\((P-1)\mu_{\mathrm{eff}}<1\), the condition number of the normalized
ROI-restricted operator is upper bounded by
\begin{equation}
\label{eq:roi_condition_bound}
\begin{aligned}
\kappa(\bar{\mathbf A}^{\mathrm{sub}})
&\le
\sqrt{
\frac{1+(P-1)\mu_{\mathrm{eff}}}
     {1-(P-1)\mu_{\mathrm{eff}}}
} \\
&=
\sqrt{
\frac{
1+\left(\frac{\mathcal R}{p_{\mathrm{prec}}}k_{\mathrm{true}}-1\right)
\mu_{\mathrm{eff}}
}{
1-\left(\frac{\mathcal R}{p_{\mathrm{prec}}}k_{\mathrm{true}}-1\right)
\mu_{\mathrm{eff}}
}
}.
\end{aligned}
\end{equation}

Equation~\eqref{eq:roi_condition_bound} directly explains the stabilizing
effect of the ROI-QP formulation. The ROI restriction reduces the number
of unknowns and tends to remove many highly coherent
background columns, thereby reducing the effective coherence
\(\mu_{\mathrm{eff}}\). Consequently, compared with the full-domain
formulation, the ROI-restricted formulation can yield a tighter
condition-number bound.

Even if the LSM-derived ROI still contains a small number of background
pixels, the effective coherence remains much smaller than that of the
full-domain operator as long as ASR-related columns dominate the reduced
subspace. Therefore, the ROI-QP step is not merely a computational
dimensionality reduction. It is consistent with the preceding
operator-coherence analysis: by suppressing the redundant background
subspace, ROI restriction improves the conditioning of the inverse problem
and stabilizes the subsequent material reconstruction.

\vspace{-1em}
\section{Numerical Simulations and Results}

\subsection{Simulation Setup}

In this section, the \ac{FDTD} method is employed to generate \ac{EM} signals for sensing the scatterers~\cite{b43}. For simplicity, we assume non-magnetic media and focus on reconstructing the dielectric constitutive parameters, where the unknown contrast $\boldsymbol\chi$ is fully characterized by the permittivity. The permeability is fixed and not estimated.
As illustrated in Fig.~\ref{fig:UCA}, a \ac{UCA} of radius $10\,\mathrm{m}$ surrounds the sensing region.
The simulations use full-wave modeling under the TM$_z$ polarization, yielding a 2-D setting that corresponds to infinitely long $z$-directed scatterers with cross-sectional permittivity contrast to be reconstructed.

The \ac{UCA} has a radius of $10\,\mathrm{m}$ and consists of $N_t=N_r=60$ antenna elements, uniformly spaced, and used for both transmission and reception.

The center frequency is set to $f_c = 28\,\mathrm{GHz}$, and frequency spacing $\Delta f = 100\,\mathrm{MHz}$.
We probe $K = 32$ frequencies $\{f_k\}_{k=1}^{K}$ and allocate $T = 8$ pilot symbols for each frequency to obtain pilot-aided channel observations.
The iterative reconstruction process is executed for $M=10$ iterations.

The sensing region $D$ is uniformly discretized into a $36 \times 36$ pixel grid centered at the origin.
In all simulations of this section, the parameters used for threshold computation of $\eta$ are fixed as $\epsilon =10^{-4}$ and $q_{\mathrm{trim}}=0.05$.
Numerical simulations are conducted using the FDTD method to obtain EM field for four cases of scatterers.

\begin{figure}[t]
  \centering
  \includegraphics[width=0.65\columnwidth]{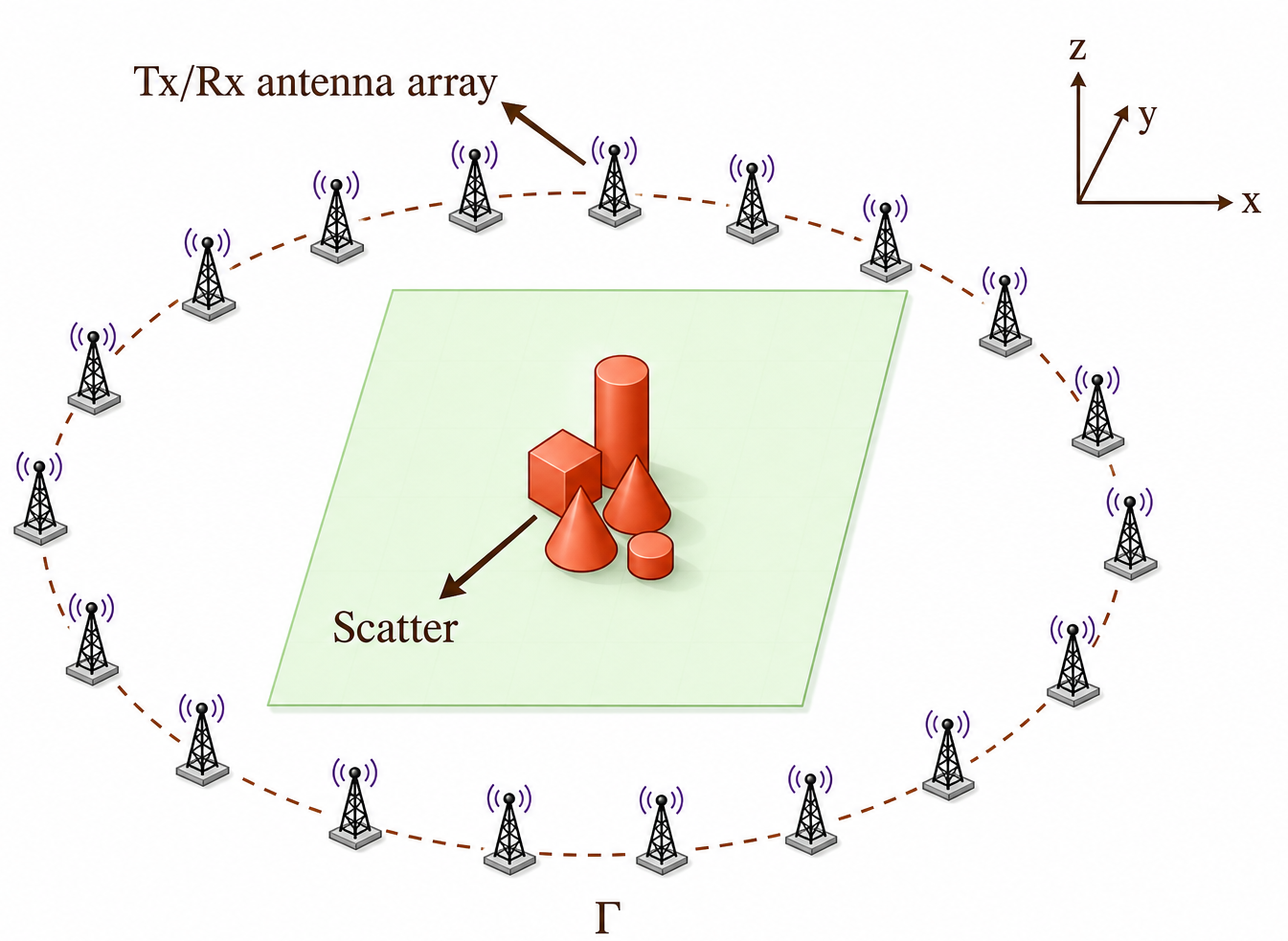}
  \caption{Schematic diagram of a simplified UCA array.}
  \label{fig:UCA}
\end{figure}
\begin{figure}[t]
  \centering
  \includegraphics[width=0.65\columnwidth]{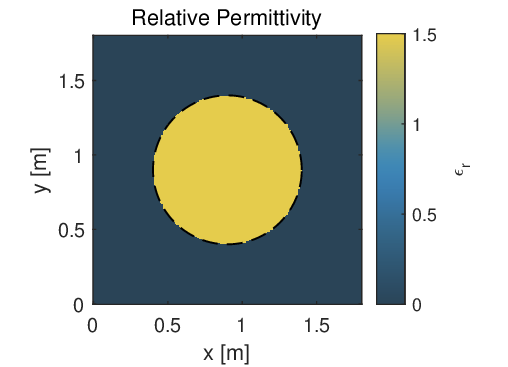}
  \caption{Geometry of the circular scatterer with a radius of $0.5\,\mathrm{m}$.}
  \label{fig:CS}
  \vspace{3mm}
\end{figure}

\vspace{-1em}
\subsection{Condition Number and Complexity}

Fig.~\ref{fig:CS} shows the simulation setup, where a $1.8$ m $\times$ $1.8$ m square sensing region contains a circular scatterer with radius $0.5$ m and relative permittivity $1.5$.
To ensure reproducibility, the \ac{ROI} sizes used in the simulations are generated deterministically rather than selected manually.
Let $\mathcal{L}_{\min}$ denote the side length of the smallest square enclosing the scatterer, and let $\mathcal{L}_{\max}$ denote the side length of the full square sensing region.
For a total of $L$ shrinking steps, the \ac{ROI} side length at step $o$ is defined as
\begin{equation}\label{eq:side_length}
\begin{aligned}
\mathcal{L}_o
&= \operatorname{round}\left(
\mathcal{L}_{\max} - \frac{o-1}{L-1}
\bigl(\mathcal{L}_{\max}-\mathcal{L}_{\min}\bigr)
\right), \\
&\hspace{-1.5em} o=1,\ldots, L ,
\end{aligned}
\end{equation}
where $\operatorname{round}(\cdot)$ denotes the nearest-integer rounding operation.
The corresponding \ac{ROI} pixel count is $P_o=\mathcal{L}_o^2$.
Therefore, once the parameter set $(\mathcal{L}_{\min},\mathcal{L}_{\max},L)$ is fixed, the entire \ac{ROI} sequence is uniquely determined.
This construction gradually reduces the reconstruction domain from the full sensing region to smaller square regions that still enclose the scatterer.

%%%%%%%%%%%%%%%%%%%%%%%%%%%%%%%%%%%%%%%%%%%
%% These are fig. 3 and 4.
\begin{figure}[t]
\centering
\includegraphics[width=0.8\columnwidth]{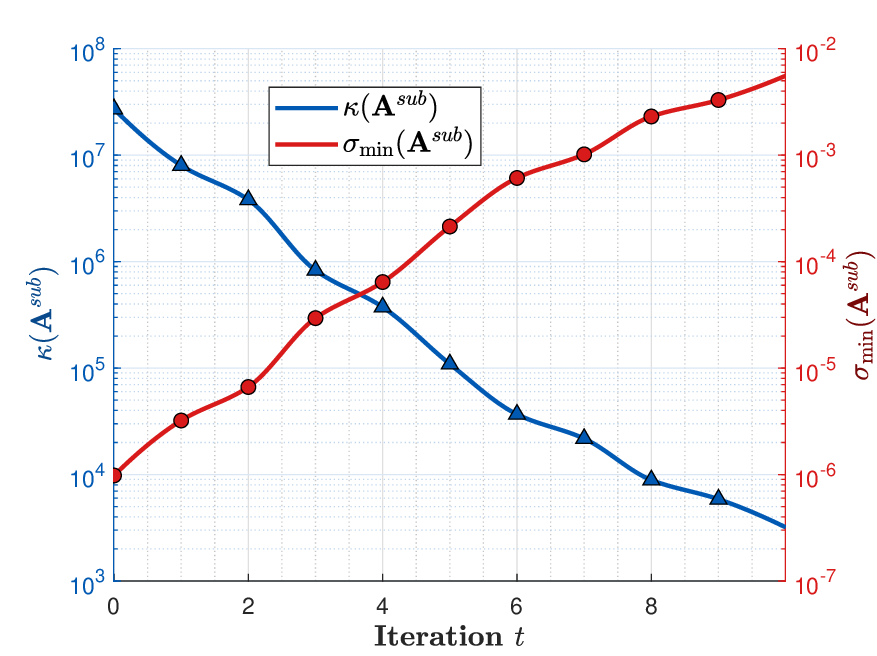}
\caption{Variation of the condition number and the smallest singular value.}
\label{fig:cond}
\end{figure}
\hspace{1\linewidth}
\begin{figure}[t]
\centering
\includegraphics[width=0.8\columnwidth]{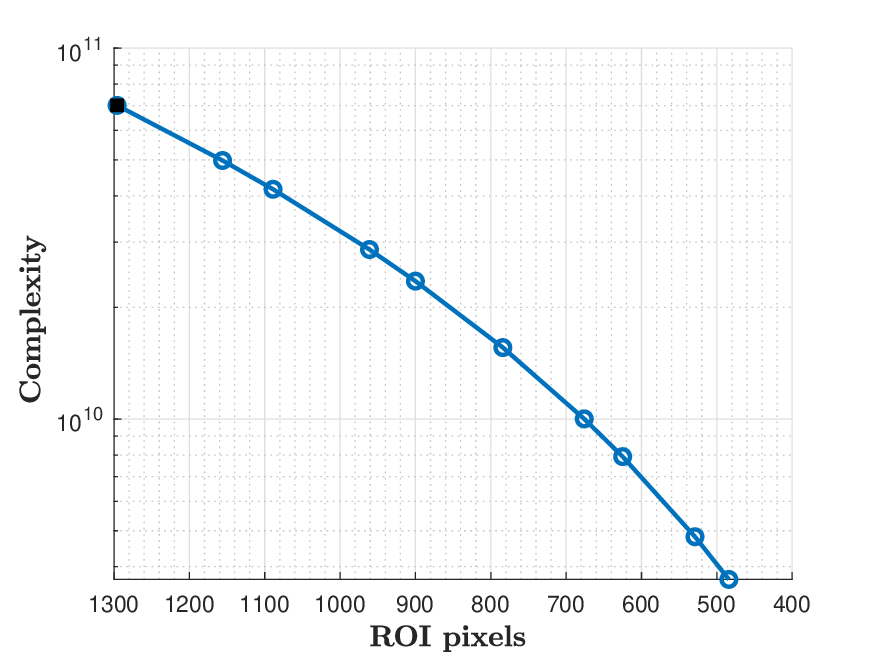}
\caption{Computational complexity versus the \ac{ROI} pixel size.}
\label{fig:complex}
\vspace{3mm}
\end{figure}
%%%%%%%%%%%%%%%%%%%%%%%%%%%%%%%%%%%%%%%%%%%

Fig.~\ref{fig:cond} evaluates the spectral effect of \ac{ROI} restriction by plotting the condition number and the smallest singular value of $\mathbf A^{\mathrm{sub}}$ during the shrinking process.
At $o=1$, the subspace operator $\mathbf A^{\mathrm{sub}}$ is identical to the full-domain operator $\mathbf A$, so no dimensional reduction is applied.
In this full-domain case, $\kappa(\mathbf A^{\mathrm{sub}})$ is on the order of $10^{7}$, while $\sigma_{\min}(\mathbf A^{\mathrm{sub}})$ is only about $10^{-6}$.
This indicates severe ill-conditioning and rapid singular-value decay.
From a spectral viewpoint, many solution-space directions are dominated by near-null modes, mainly caused by highly coherent air-region columns.
These modes amplify noise and lead to unstable \ac{CPR} estimates.

As the \ac{ROI} shrinks, redundant air-region columns are progressively removed from $\mathbf A^{\mathrm{sub}}$.
The remaining operator is therefore increasingly governed by \ac{ASR} columns, which carry the intrinsic \ac{EM} information of the scatterer.
This selective column removal improves the effective rank and preserves more singular values with non-negligible magnitudes.
Consequently, $\sigma_{\min}(\mathbf A^{\mathrm{sub}})$ increases from approximately $10^{-6}$ to $10^{-2}$, while $\kappa(\mathbf A^{\mathrm{sub}})$ decreases from $10^{7}$ to about $10^{3}$.
These results verify that \ac{ROI} restriction stabilizes the inverse problem by suppressing noise-sensitive background components.

Fig.~\ref{fig:complex} further shows the computational benefit of the same \ac{ROI} restriction.
As the reconstruction domain is reduced, the number of active pixels decreases from $1,296$ in the full domain to $497$.
Accordingly, the algorithmic complexity is reduced by nearly one order of magnitude.
Therefore, the \ac{ROI} constraint improves both the numerical conditioning and the computational efficiency of \ac{CPR}, since the reconstruction is focused on the physically meaningful scatterer region rather than the redundant air background.

\subsection{CPR Results}

To further assess the performance of the proposed method, we present
the \ac{CPR} results for dielectric contrast reconstruction under
different scenarios. The numerical evaluation is organized as follows.
First, we consider a single convex scatterer to verify the basic
reconstruction capability of the proposed \ac{ROI}-\ac{QP} framework.
Second, we examine a two-scatterer elliptical cluster to evaluate the
separation capability under closely spaced scatterers. Finally, we use
a non-convex T-shaped scatterer as a representative ablation case to
compare the full-domain \ac{BIM}, the \ac{LSM}-guided \ac{ROI}-\ac{BIM},
and the proposed \ac{LSM}-guided \ac{ROI}-\ac{QP}.

Specifically, our comparisons are designed to probe the impact of the
scattering operator in the \ac{ROI}-restricted linear
model~\eqref{eq:subproblem}, rather than to establish state-of-the-art
performance across different inversion paradigms. To this end, we employ
a conventional \ac{BIM} with Tikhonov regularization as a reference
implementation of the linearized update, because its reconstruction
behavior is primarily dictated by the conditioning of the associated
scattering operator. Unless otherwise specified, all methods are evaluated
under an identical forward model, the same noisy observation, and the
same number of outer Born iterations. The remaining simulation parameters
follow Section~VI-A.

\subsubsection{Single Scatterer: Triangular Case}

\begin{figure}[t]
  \centering

  % Row 1
  \subfloat[]{%
    \includegraphics[width=0.45\columnwidth]{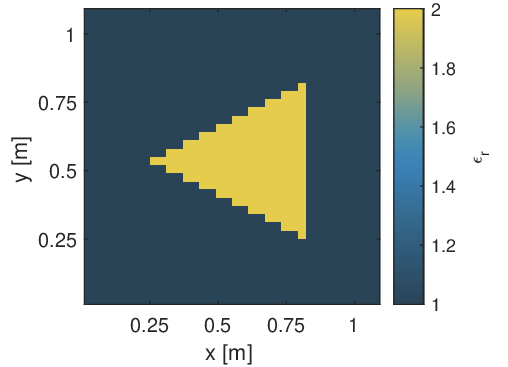}%
  }\hspace{-0.02\columnwidth}%
  \subfloat[]{%
    \includegraphics[width=0.45\columnwidth]{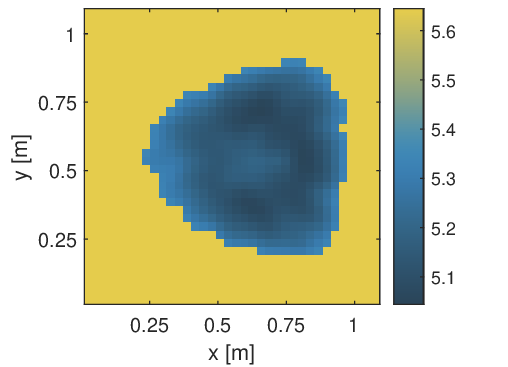}%
  }

  \vspace{-0.8em}

  % Row 2
  \subfloat[]{%
    \includegraphics[width=0.45\columnwidth]{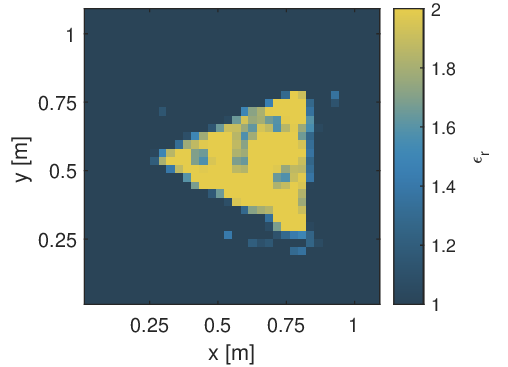}%
  }\hspace{-0.02\columnwidth}%
  \subfloat[]{%
    \includegraphics[width=0.45\columnwidth]{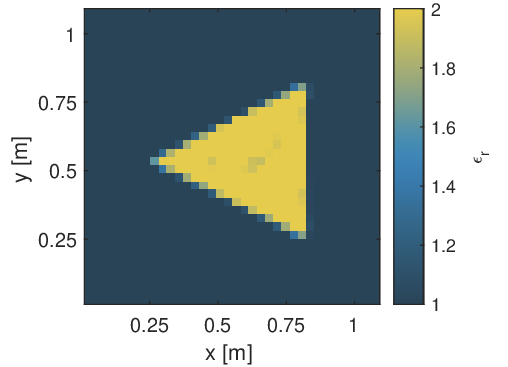}%
  }

  \vspace{-0.8em}

  % Row 3
  \subfloat[]{%
    \includegraphics[width=0.45\columnwidth]{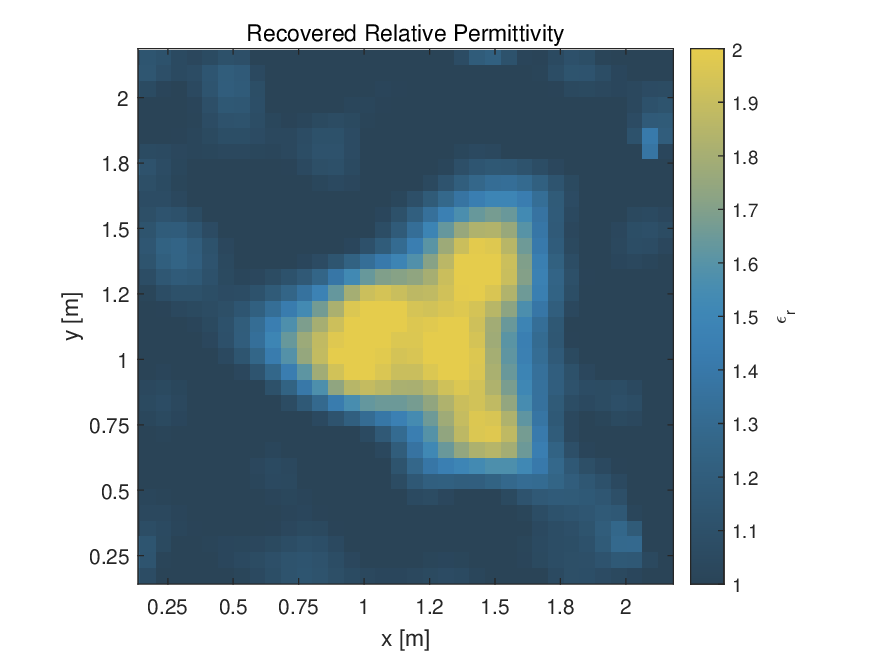}%
  }\hspace{-0.02\columnwidth}%
  \subfloat[]{%
    \includegraphics[width=0.45\columnwidth]{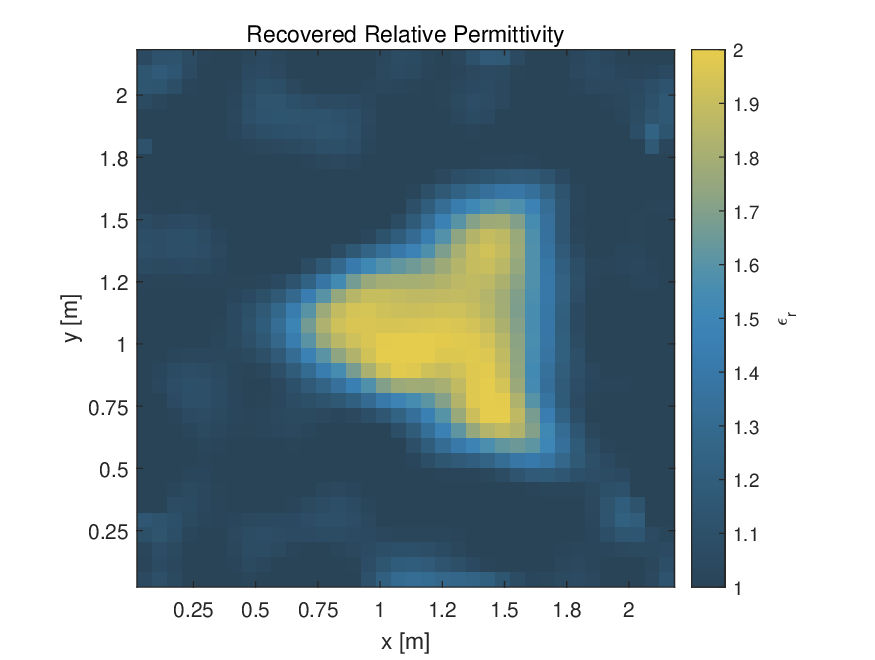}%
  }

  \caption{Single triangular scatterer reconstruction.
  (a) True permittivity; (b) \ac{ROI} obtained after truncation;
  (c)--(d) \ac{ROI}-\ac{QP} reconstructions with single-step Born
  linearization at $5$~dB and $30$~dB; (e)--(f) \ac{BIM}
  reconstructions at $5$~dB and $30$~dB.}
  \label{fig:fiveplots}
\end{figure}

The permittivity of the triangular scatterer is set to $2$.
Fig.~\ref{fig:fiveplots}(a) depicts an equilateral triangular scatterer
with an area of $0.65\,\mathrm{m}^2$. We first apply the \ac{LSM} under
$5$~dB noise with the regularization parameter $\zeta=10^{-3}$. After
thresholding, the \ac{ROI} shown in Fig.~\ref{fig:fiveplots}(b) is
extracted, which reduces the problem dimension from $1{,}296$ to $426$
pixels.

Within this \ac{ROI}, we run the \ac{ROI}-constrained \ac{QP} under both
$5$~dB and $30$~dB \acp{SNR}, as shown in Figs.~\ref{fig:fiveplots}(c)
and~\ref{fig:fiveplots}(d), respectively. In both cases, the proposed
\ac{ROI}-\ac{QP} successfully recovers the scatterer geometry, while the
higher-\ac{SNR} case further improves the accuracy of the reconstructed
internal permittivity. For comparison, Figs.~\ref{fig:fiveplots}(e)
and~\ref{fig:fiveplots}(f) report the corresponding \ac{BIM}
reconstructions. Under the same simulation setting, the
\ac{ROI}-constrained \ac{QP} yields more accurate \ac{CPR} results and
fewer background artifacts.

\subsubsection{Multiple Scatterers: Elliptical Cluster Case}

\begin{figure}[t]
  \centering

  % Row 1
  \subfloat[]{%
    \includegraphics[width=0.45\columnwidth]{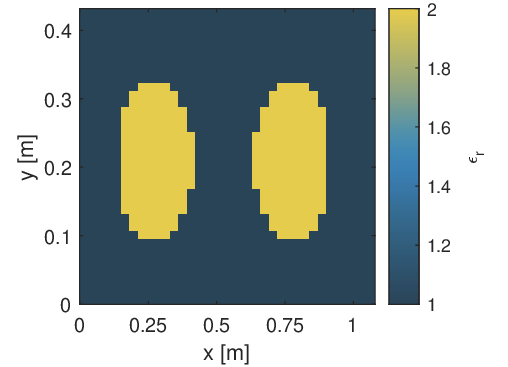}%
  }\hspace{-0.02\columnwidth}%
  \subfloat[]{%
    \includegraphics[width=0.45\columnwidth]{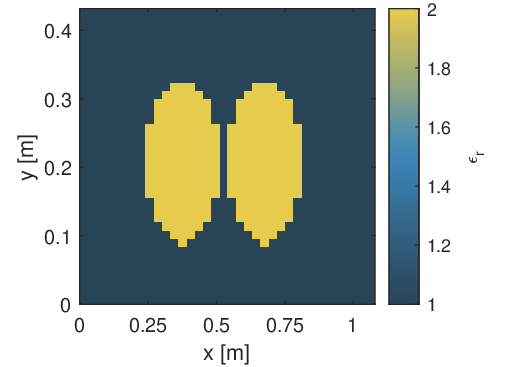}%
  }

  \vspace{-0.8em}

  % Row 2
  \subfloat[]{%
    \includegraphics[width=0.45\columnwidth]{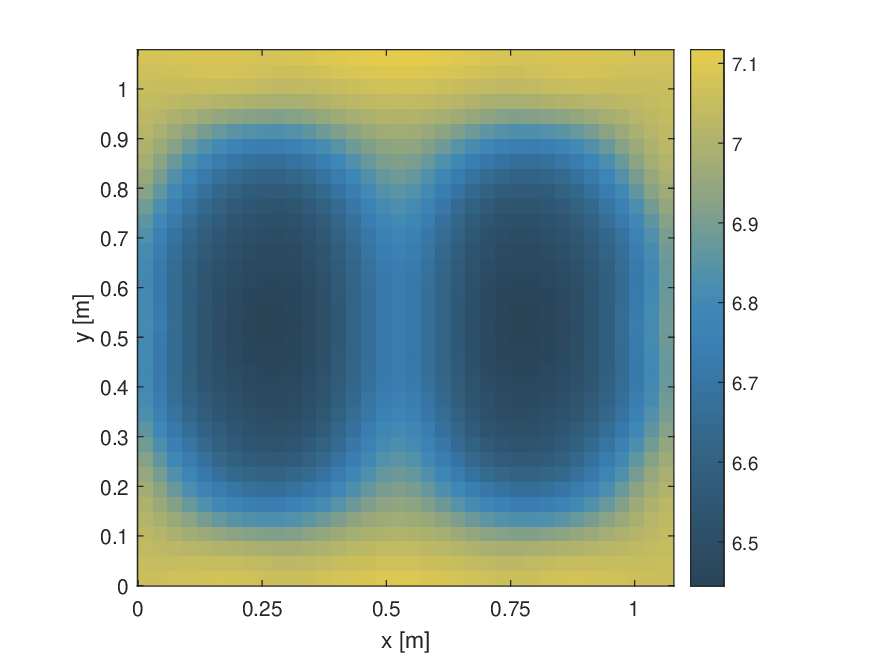}%
  }\hspace{-0.02\columnwidth}%
  \subfloat[]{%
    \includegraphics[width=0.45\columnwidth]{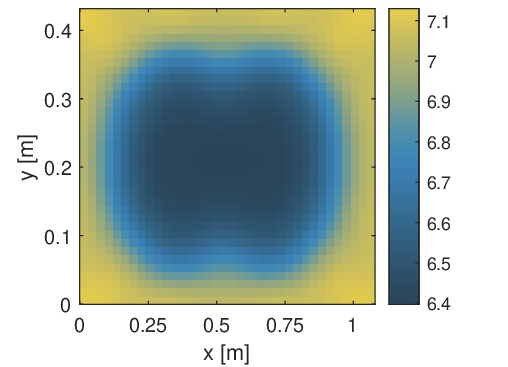}%
  }

  \vspace{-0.8em}

  % Row 3
  \subfloat[]{%
    \includegraphics[width=0.45\columnwidth]{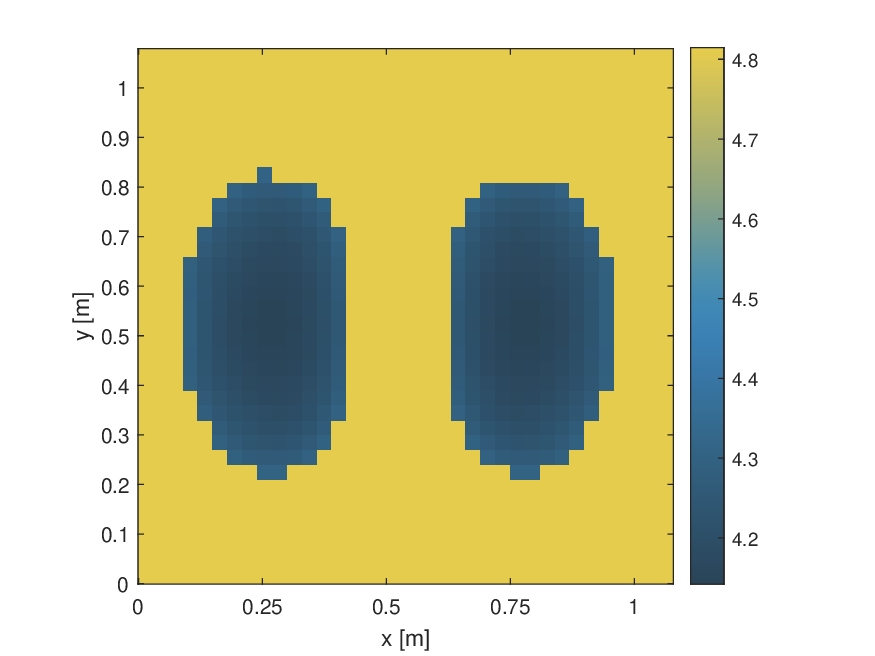}%
  }\hspace{-0.02\columnwidth}%
  \subfloat[]{%
    \includegraphics[width=0.45\columnwidth]{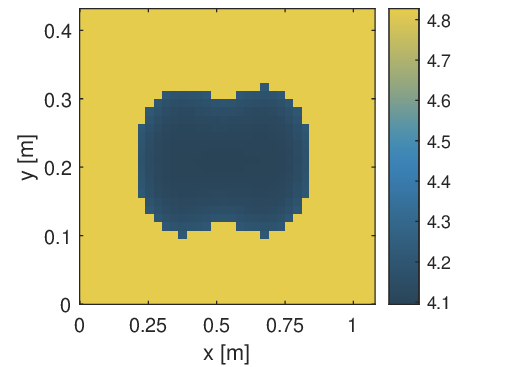}%
  }

  \vspace{-0.8em}

  % Row 4
  \subfloat[]{%
    \includegraphics[width=0.45\columnwidth]{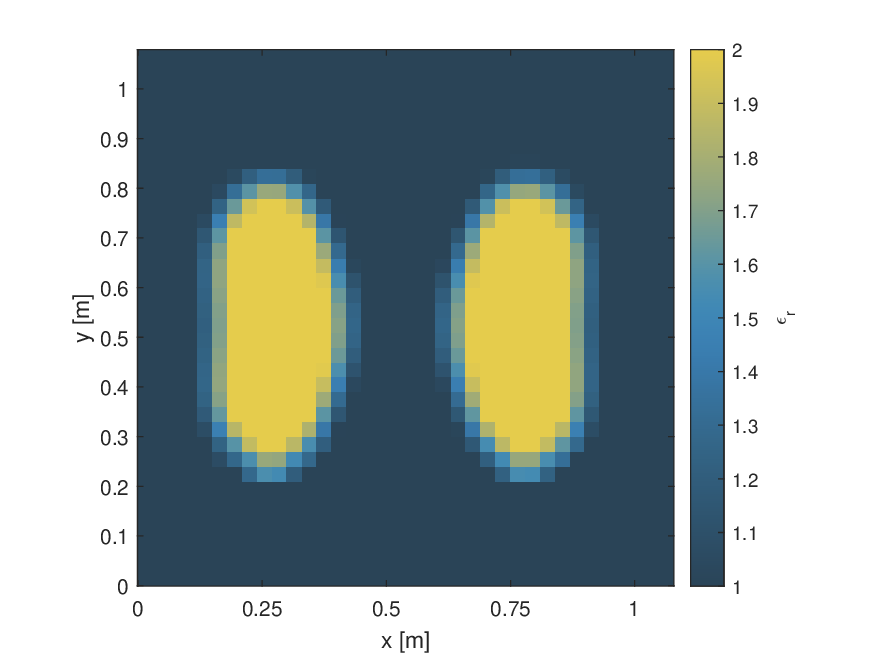}%
  }\hspace{-0.02\columnwidth}%
  \subfloat[]{%
    \includegraphics[width=0.45\columnwidth]{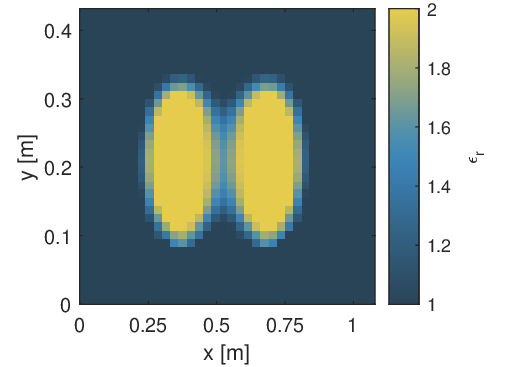}%
  }

  \caption{Two-ellipse scatterer reconstruction at $10$~dB.
  Left column: $0.3$~m spacing; right column: $0.1$~m spacing.
  (a)--(b) True permittivity distributions; (c)--(d) \ac{LSM}
  indicators; (e)--(f) truncated \acp{ROI}; (g)--(h)
  \ac{ROI}-\ac{QP} reconstructions.}
  \label{fig:8plots}
\end{figure}

Next, we evaluate the performance of the \ac{ROI}-\ac{QP} scheme for
a cluster of two closely spaced elliptical scatterers. The scatterers
have relative permittivity $2$. Each ellipse has a major axis of
$0.12\,\mathrm{m}$ and a minor axis of $0.055\,\mathrm{m}$, and the
center-to-center spacing is set to $0.3$~m and $0.1$~m, respectively.
The reconstructions are evaluated at $10$~dB \ac{SNR}.

The corresponding \ac{LSM} indicators with $\zeta=10^{-4}$ are shown
in Figs.~\ref{fig:8plots}(c) and~\ref{fig:8plots}(d). When applied to
a cluster of scatterers, the \ac{LSM} exhibits a noticeable
\textit{sticking} effect along the boundaries of adjacent objects. This
phenomenon arises because the \ac{LSM} tends to produce spatially
extended high-response regions around strong permittivity boundaries.
As the separation between scatterers approaches the resolution limit of
the sensing aperture, these regions overlap and merge, resulting in
apparent boundary adhesion and reduced object separability.

After applying the truncation rule, the \acp{ROI} corresponding to the
two spacing cases are obtained, as shown in Figs.~\ref{fig:8plots}(e)
and~\ref{fig:8plots}(f). For the $0.3$~m spacing, the \ac{ROI}
successfully encloses both ellipses within a moderately expanded
boundary, providing accurate geometric localization. In contrast, for
the $0.1$~m spacing, the \ac{ROI} can no longer fully resolve the two
ellipses, and partial adhesion occurs between their boundaries.

The \ac{ROI}-\ac{QP} reconstructions are shown in
Figs.~\ref{fig:8plots}(g) and~\ref{fig:8plots}(h). The proposed method
recovers both the shapes and permittivity values of the scatterers for
both inter-scatterer spacings, with only minor artifacts near the
boundaries. Even when the \ac{LSM}-derived \ac{ROI} suffers from boundary
adhesion at $0.1$~m spacing, this issue is largely alleviated during the
\ac{ROI}-\ac{QP} refinement. This improvement stems from the
data-consistency enforcement inherent in the Born iteration, which
continuously adjusts the reconstructed permittivity to minimize the
residual between the simulated and measured scattering fields. Through
this feedback mechanism, the algorithm adaptively distinguishes subtle
dielectric variations between the closely spaced scatterers, thereby
restoring their individual shapes and improving reconstruction fidelity.

Overall, these results highlight that the proposed \ac{ROI}-\ac{QP}
approach maintains robust separation capability and stable permittivity
recovery even in highly coupled and relatively low-\ac{SNR} environments,
where qualitative localization by the \ac{LSM} alone may become
ambiguous.

\subsubsection{Non-Convex Scatterer: T-Shaped Ablation Case}

\begin{figure}[t]
  \centering
  \hspace*{-1.8em}
  \begin{overpic}[
    width=1.1\columnwidth,
    trim={0 63 0 80},
    clip
  ]{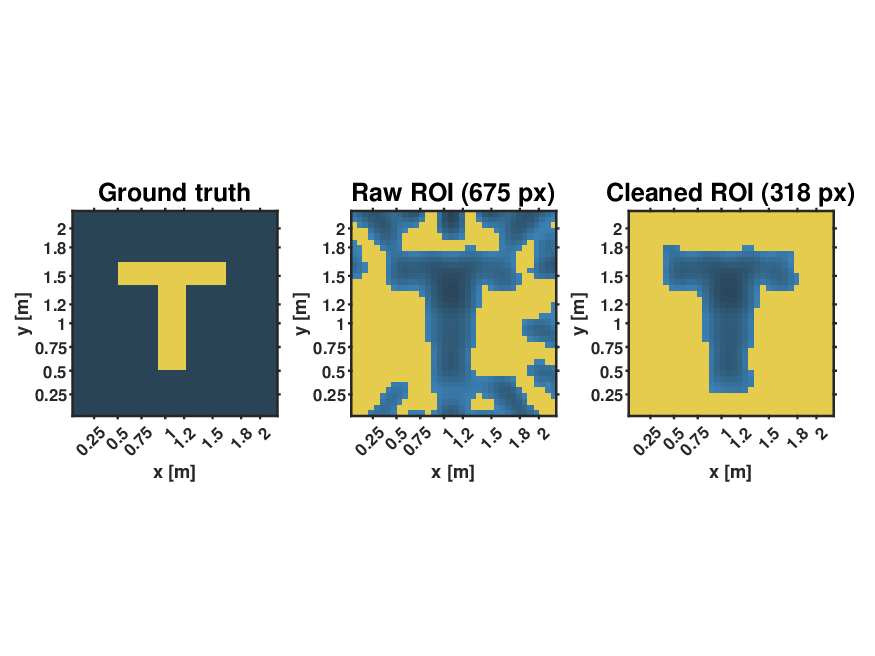}
    \put(18,1){\small (a)}
    \put(50,1){\small (b)}
    \put(82,1){\small (c)}
  \end{overpic}

  \caption{Cleaned LSM-derived ROI for the T-shaped scatterer:
  (a) ground truth; (b) raw \ac{ROI}; (c) cleaned \ac{ROI}.}
  \label{fig:tshape_lsm}
\end{figure}

\begin{figure}[t]
  \centering
  \hspace*{-1.1em}
  \begin{overpic}[
    width=1.05\columnwidth,
    trim={0 63 0 80},
    clip
  ]{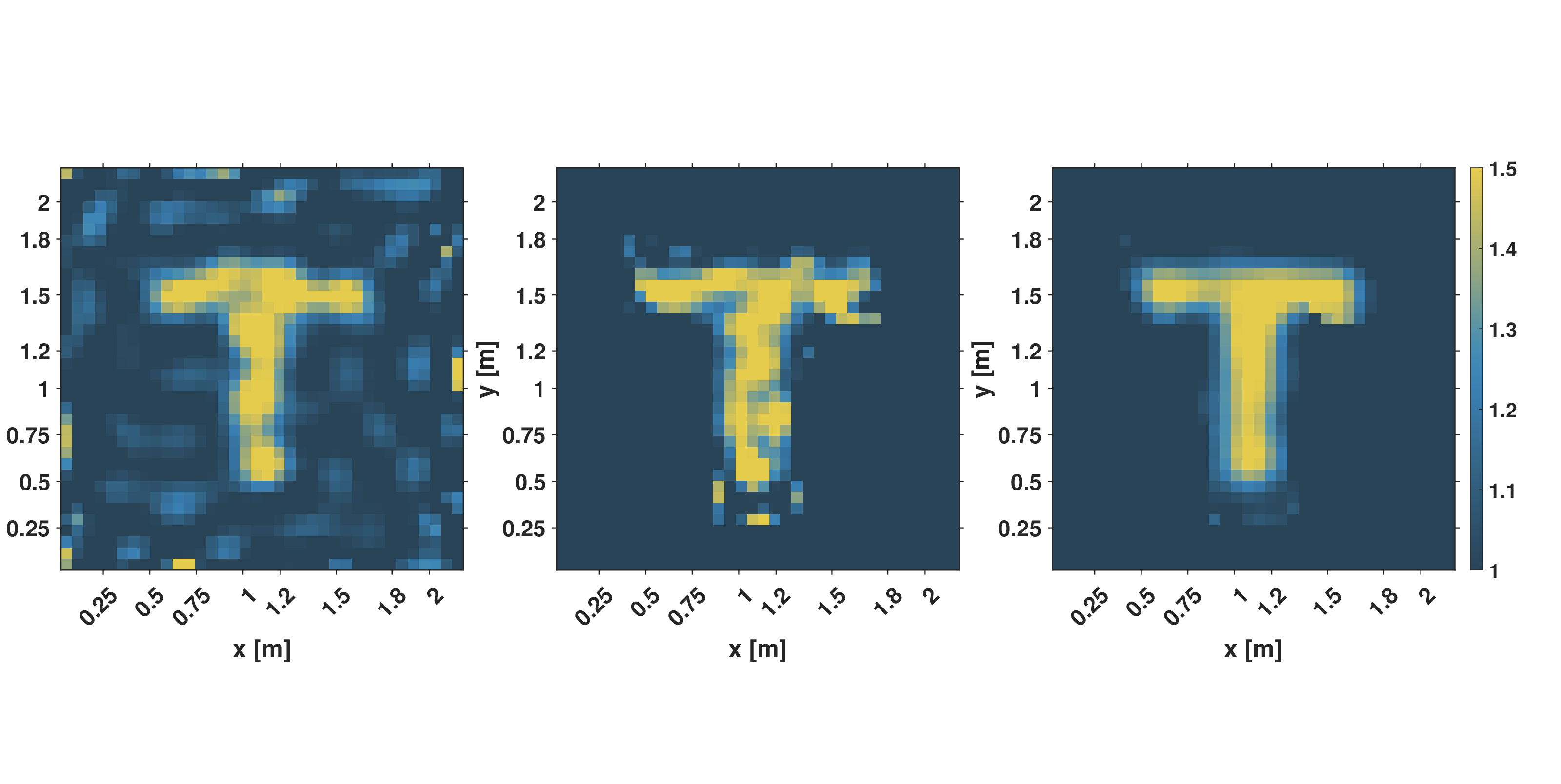}
    \put(16,0){\small (a)}
    \put(47,0){\small (b)}
    \put(79,0){\small (c)}
  \end{overpic}
  \caption{\ac{CPR} results of: (a) full-domain \ac{BIM}; (b)
  \ac{LSM}-guided \ac{ROI}-\ac{BIM}; (c) proposed \ac{ROI}-\ac{QP}.}
  \label{fig:tshape_cpr}
\end{figure}

\begin{figure}[t]
  \centering
  \begin{overpic}[
    width=1.1\columnwidth,
    trim={0 0 0 0},
    clip
  ]{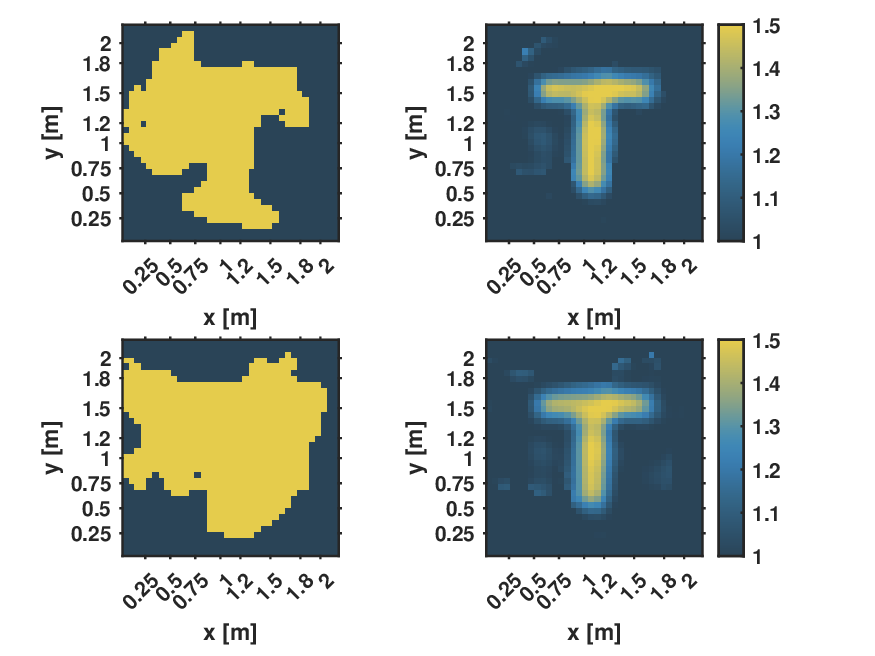}
    \put(24,75){\colorbox{white}{\small (a)}}
    \put(66,75){\colorbox{white}{\small (b)}}
    \put(24,-2){\small (c)}
    \put(66,-2){\small (d)}
  \end{overpic}
  \caption{Extended \ac{ROI} and corresponding \ac{CPR} results:
  (a) extended \ac{ROI} with $550$ pixels; (b) \ac{CPR} result of the
  $550$-pixel \ac{ROI}; (c) extended \ac{ROI} with $700$ pixels;
  (d) \ac{CPR} result of the $700$-pixel \ac{ROI}.}
  \vspace{2mm}
  \label{fig:roi_extendD}
\end{figure}

We further evaluate the proposed method on a non-convex T-shaped scatterer.
As shown in Fig.~\ref{fig:tshape_lsm}(a), the target has an arm length of $1.1$~m along both axes, a width of $0.24$~m, and a relative permittivity of $1.5$, while the background relative permittivity is $1.0$.
This case is challenging because the non-convex geometry and low-\ac{SNR} measurements can lead to ambiguous \ac{LSM}-based localization.

Fig.~\ref{fig:tshape_lsm}(b) shows the raw \ac{ROI} obtained from the \ac{LSM} indicator.
Due to noise and the non-convex target shape, the raw \ac{ROI} contains boundary artifacts and isolated clutter pixels.
Directly including these pixels in the inversion subspace may introduce irrelevant degrees of freedom and degrade the conditioning of the linearized inverse problem.
Therefore, after trimmed maximum-gap thresholding, we apply a simple \ac{ROI} refinement step. 
The raw \ac{LSM} \ac{ROI} is defined as
\begin{equation}
\Omega_{\rm raw}
=
{\mathbf r:\mathcal J_K(\mathbf r)\le \eta},
\end{equation}
where $\eta$ is determined by the trimmed maximum-gap rule applied to the values of $\mathcal J_K(\mathbf r)$ over the discretized sensing domain.
Since the \ac{LSM} solution norm is typically smaller inside the scatterer support, a smaller value of $\mathcal J_K(\mathbf r)$ indicates a more scatterer-like location.

We then retain the most reliable low-indicator region
\begin{equation}
\Omega_{\rm core}
=
{\mathbf r\in\Omega_{\rm raw}:
\mathcal J_K(\mathbf r)\le \tau_{\rm core}},
\end{equation}
where $\tau_{\rm core}= 0.4$ is selected as a fixed percentile threshold of $\mathcal J_K(\mathbf r)$ within $\Omega_{\rm raw}$.
Small isolated components and boundary-connected artifacts are removed, the main connected component is slightly expanded within the original \ac{LSM} support, and
interior holes are filled.
This refinement uses only the \ac{LSM} indicator and the raw \ac{ROI}, without using ground-truth contrast values or target labels.
The cleaned \ac{ROI} is shown in Fig.~\ref{fig:tshape_lsm}(c), which preserves the main T-shaped support while suppressing most irrelevant artifacts.

Using this cleaned \ac{ROI}, we compare the proposed method with two representative baselines at $0$~dB \ac{SNR}: full-domain \ac{BIM} with $\ell_2$ regularization and \ac{LSM}-guided \ac{ROI}-\ac{BIM}.
All methods use the same \ac{CSI}-induced forward model, noisy observations, and number of outer Born iterations, so that the comparison mainly reflects the effects of \ac{ROI} restriction and the proposed \ac{QP} formulation.

The reconstruction results are shown in Fig.~\ref{fig:tshape_cpr}. 
Compared with the full-domain \ac{BIM} in Fig.~\ref{fig:tshape_cpr}(a), the \ac{LSM}-guided \ac{ROI}-\ac{BIM} in Fig.~\ref{fig:tshape_cpr}(b) reduces background artifacts and improves the contrast profile.
This agrees with the conditioning analysis: removing irrelevant background pixels reduces the effective dimension of the inverse problem and suppresses noise-driven artifacts.
As shown in Fig.~\ref{fig:tshape_cpr}(c), the proposed \ac{ROI}-\ac{QP} further produces a cleaner T-shaped reconstruction and better recovers the target contrast value.
This improvement comes from the joint effect of \ac{LSM}-guided \ac{ROI} restriction and the regularized \ac{QP} formulation, which balances data consistency, Tikhonov regularization, and graph-Laplacian smoothness over the \ac{ROI}.

For quantitative comparison, we use the \ac{NMSE} of the reconstructed contrast vector,
\begin{equation}
\label{NMSE}
\mathrm{NMSE}
=
10 \log_{10}
\left(
\frac{\|\hat{\boldsymbol{\chi}}-\boldsymbol{\chi}\|_2^2}
{\|\boldsymbol{\chi}\|_2^2}
\right).
\end{equation}
We also report the numerical rank of the sensing operator, defined as
\begin{equation}
\label{eq:numerical_rank}
r_{\rm num}
\triangleq
\left|\left\{
i:\sigma_i \ge \epsilon_{\rm rank}\sigma_{\max}
\right\}\right|,
\end{equation}
where $\sigma_i$ denotes the $i$th singular value and
$\epsilon_{\rm rank}=10^{-6}$.

Table~\ref{tab:tshape-performance} summarizes the comparison.
The full-domain \ac{BIM} has the largest numerical rank, since it solves the inverse problem over the entire sensing grid. After \ac{ROI} restriction, the numerical rank is reduced from $522$ to $174$ for both \ac{ROI}-\ac{BIM} and \ac{ROI}-\ac{QP}, indicating a lower-dimensional inverse problem. 
In terms of reconstruction accuracy, \ac{ROI}-\ac{BIM} improves the \ac{NMSE} from $-6.10$~dB to $-7.27$~dB, while the proposed \ac{ROI}-\ac{QP} further improves it to $-9.38$~dB.

\begin{table}[t]
\centering
\caption{Comparison with representative inverse-scattering baselines for
the T-shaped scatterer.}
\label{tab:tshape-performance}
\hspace*{-1.4em}
\begin{tabular}{lccc}
\hline
Method & Numerical rank & NMSE (dB) & Running time (s) \\
\hline
Full-domain BIM & 522 & $-6.10$ & $91.09$ \\
LSM-guided ROI-BIM & 174 & $-7.27$ & $17.38 + 0.59$ \\
Proposed ROI-QP & 174 & $-9.38$ & $15.56 + 0.59$ \\
\hline
\end{tabular}
\vspace{5mm}
\end{table}

We also report the absolute running time as a numerical reference. 
Because the original $10$~m \ac{UCA} setting makes the \ac{FDTD} simulation time for this case excessively long, the time comparison is conducted under a reduced setting with a $0.5$~m \ac{UCA} radius and adjusted \ac{FDTD} parameters.
This reduced setting is used only for the absolute-time comparison.
The simulations were performed in MATLAB R2024b on a computer with a $16$-core
12th Gen Intel Core i7-12650H CPU.
The full-domain \ac{BIM} takes $91.09$~s, whereas \ac{LSM}-guided \ac{ROI}-\ac{BIM} and the proposed \ac{ROI}-\ac{QP} take $17.38+0.59$~s and $15.56+0.59$~s, respectively, where
$0.59$~s is the \ac{LSM}-based \ac{ROI} selection time.
Therefore, these absolute values are used only as a reference for computational efficiency under the reduced-radius and adjusted-grid setting.

Finally, we test the sensitivity of \ac{ROI}-\ac{QP} to \ac{ROI} mismatch.
Starting from the cleaned \ac{ROI} in Fig.~\ref{fig:tshape_lsm}(c), we randomly and smoothly enlarge the \ac{ROI} from $318$ pixels to $550$ and $700$ pixels under $\mathrm{SNR}=5$~dB, as shown in Fig.~\ref{fig:roi_extendD}.
The enlarged \acp{ROI} still cover the target but include different amounts of background pixels.
The proposed \ac{ROI}-\ac{QP} remains effective, achieving \ac{NMSE} values of $-8.12$~dB and $-7.56$~dB for the $550$-pixel and $700$-pixel \acp{ROI}, respectively.
The degradation for the larger \ac{ROI} is consistent with our analysis, since additional background pixels increase the effective dimension and may introduce edge artifacts.

\subsection{Normalized Mean Square Error Analysis}

In this subsection, we validate the above conclusions by examining how
the \ac{NMSE} of all \ac{CPR} simulation results varies with the
\ac{ROI} size and the noise level in dB; see
Figs.~\ref{fig:nmse-tri}--\ref{fig:nmse-ell}. In particular, when
investigating the \ac{NMSE}--\ac{ROI}-pixel trend, we fix the number of
\ac{QP} iterations to one, i.e., a single Born step, to control for any
additional performance gain introduced by the \ac{QP} refinement itself,
so that the observed variation mainly reflects the improvement brought by
the \ac{ROI} constraint. In contrast, for the \ac{NMSE}--\ac{SNR}
results, the \ac{QP} refinement is executed with the nominal iteration
setting to report the performance of the full reconstruction pipeline.

\begin{figure}[t]
  \centering
  \subfloat[]{%
    \includegraphics[width=0.75\columnwidth,trim=10 0 10 8,clip]{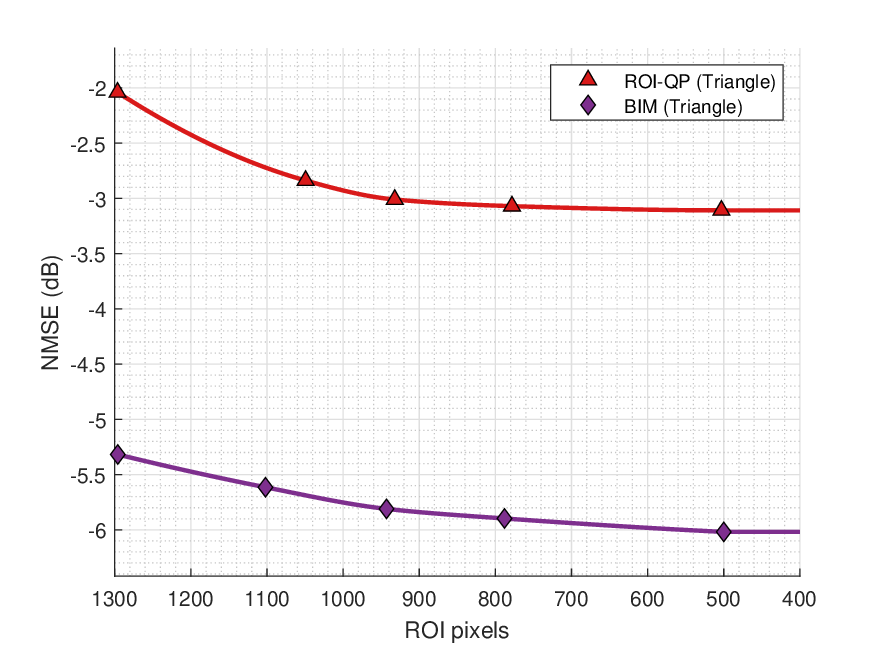}%
  }\hfill
  \subfloat[]{%
    \includegraphics[width=0.75\columnwidth,trim=10 0 10 8,clip]{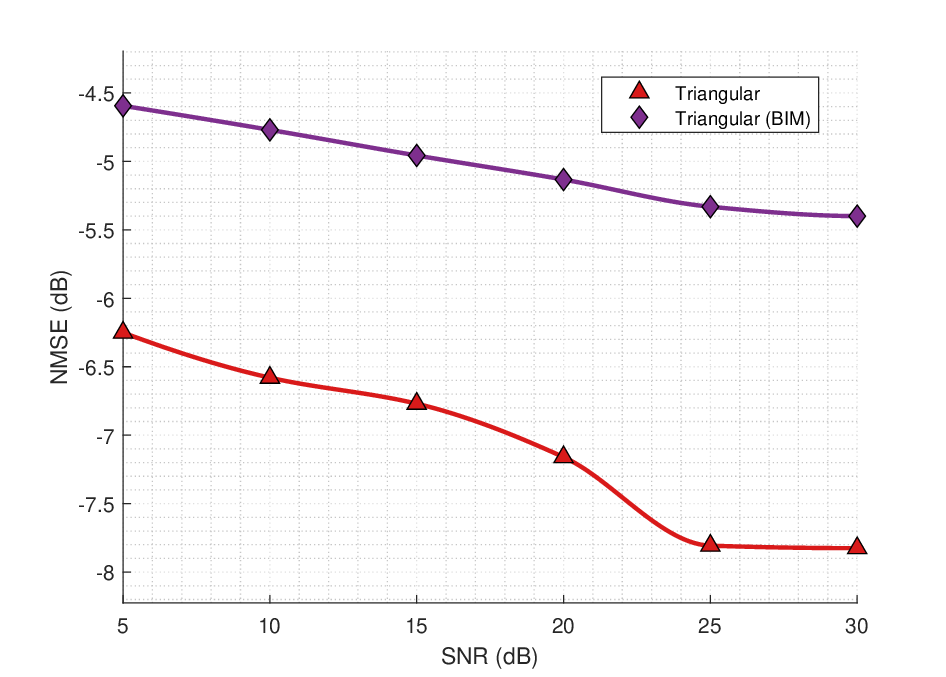}%
  }
  \caption{NMSE results for the triangular scatterer: (a) NMSE versus ROI pixels, and (b) NMSE versus SNR.}
  \label{fig:nmse-tri}
  \vspace{3mm}
\end{figure}

\begin{figure}[t]
  \centering
  \subfloat[]{%
    \includegraphics[width=0.75\columnwidth,trim=10 0 10 8,clip]{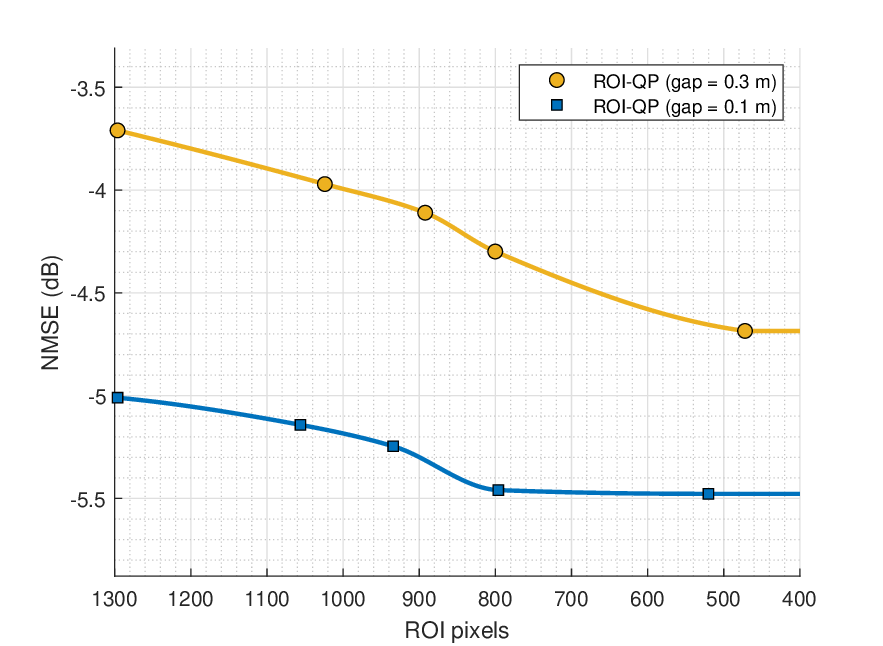}%
  }\hfill
  \subfloat[]{%
    \includegraphics[width=0.75\columnwidth,trim=10 0 10 8,clip]{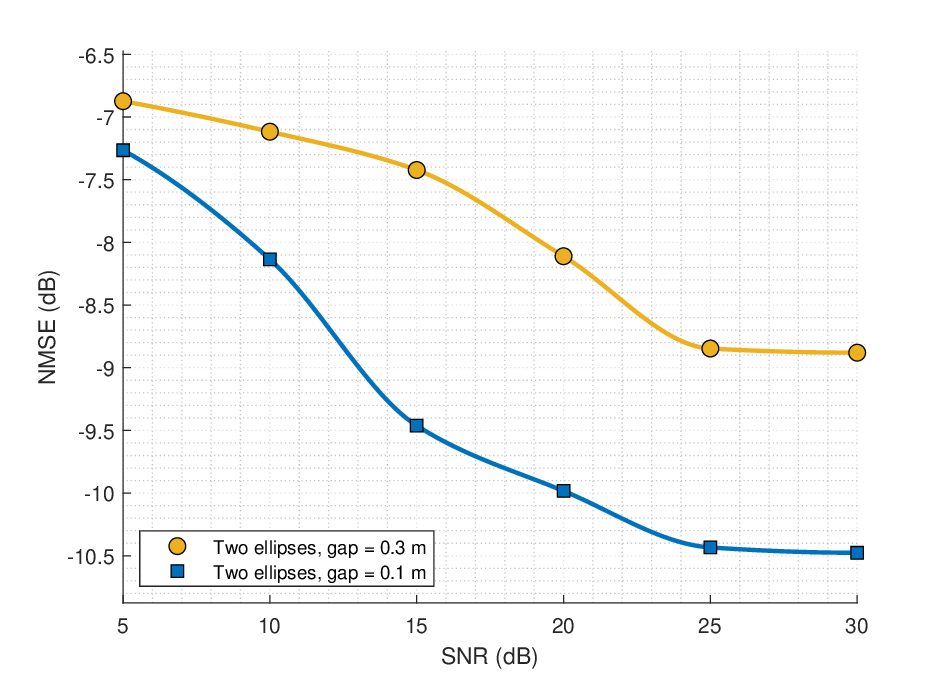}%
  }
  \caption{NMSE results for the two-ellipse cases: (a) NMSE versus ROI pixels, and (b) NMSE versus SNR.}
  \label{fig:nmse-ell}
  \vspace{3mm}
\end{figure}

From Fig.~\ref{fig:nmse-tri}(a), we observe that, at
$\mathrm{SNR}=30$~dB, the \ac{NMSE} of both \ac{ROI}-\ac{QP} and
Tikhonov-\ac{BIM} decreases as the \ac{ROI} shrinks from the full domain
($1{,}296$ pixels) toward the \ac{ASR}, confirming the benefit of removing
redundant background pixels. The \ac{NMSE} reduction is about $1$~dB for
\ac{ROI}-\ac{QP} in the triangular case, while the triangular case with
Tikhonov-\ac{BIM} improves by about $0.7$~dB. Note that the
\ac{NMSE}--\ac{ROI} curves are obtained with a single \ac{QP} iteration
(one-step Born update) to exclude the gain from iterative refinement;
consequently, \ac{ROI}-\ac{QP} appears roughly $3$~dB worse than
Tikhonov-\ac{BIM} in this controlled setting.

When the nominal iterations are enabled, as shown in
Fig.~\ref{fig:nmse-tri}(b), \ac{ROI}-\ac{QP} achieves higher accuracy,
and both methods exhibit the expected monotonic \ac{NMSE} decrease with
increasing \ac{SNR}. A similar decreasing trend with respect to the
\ac{ROI} size is also observed for the two-ellipse cases in
Fig.~\ref{fig:nmse-ell}(a), and the corresponding \ac{NMSE}--\ac{SNR}
behavior is reported in Fig.~\ref{fig:nmse-ell}(b), where the smaller
gap, $0.1$~m, yields a lower initial \ac{NMSE} and a milder sensitivity
to \ac{ROI} reduction. From Fig.~\ref{fig:nmse-t}(a), we observe that,
at $\mathrm{SNR}=30$~dB, the \ac{NMSE} of both \ac{ROI}-\ac{QP} and
Tikhonov-\ac{BIM} decreases as the \ac{ROI} shrinks from the full domain
($1{,}296$ pixels) toward the \ac{ASR}, confirming the benefit of removing
redundant background pixels. The \ac{NMSE} reduction is about $1$~dB for
both \ac{ROI}-\ac{QP} and Tikhonov-\ac{BIM} in the T-shaped case. When
the nominal iterations are enabled, as shown in Fig.~\ref{fig:nmse-t}(b),
\ac{ROI}-\ac{QP} achieves higher accuracy, and both methods exhibit the
expected monotonic \ac{NMSE} decrease with increasing \ac{SNR}.

In summary, the \ac{NMSE} trends of both \ac{ROI}-\ac{QP} and \ac{BIM}
with respect to the \ac{ROI} pixel number corroborate our theoretical
analysis. Moreover, under the nominal iteration setting, \ac{ROI}-\ac{QP}
achieves improved reconstruction accuracy across different scatterer
geometries and noise conditions.

\section{Conclusion}
This paper investigated the origins of ill-conditioning in \ac{CSI}-based electromagnetic material reconstruction for \ac{ISAC} systems.
From an operator-centric viewpoint, we identified a structural dichotomy in the \ac{CSI}-induced \ac{ISAC} scattering operator: background-related columns are highly coherent and dominate the near-rank deficiency, whereas scatterer-related columns are comparatively weakly correlated and determine the effective rank.
This observation provides a theoretical rationale for \ac{ROI}-restricted inversion as a principled means of stabilizing the inverse problem. 
We further derived analytical condition-number bounds that connect the \ac{ROI} size, effective coherence, and conditioning improvement, thereby quantifying the benefit of subspace restriction. 
The proposed analysis not only explains the ill-conditioning inherent in \ac{CSI}-driven material reconstruction but also provides practically relevant theoretical guidance for designing improved algorithms that infer environmental scattering information from communication-native data.
In future environment-aware \ac{ISAC} systems, this framework may support applications such as channel prediction, radio-environment mapping, and wireless digital-twin construction.

\begin{figure}[t]
  \centering
  \subfloat[]{%
    \includegraphics[width=0.75\columnwidth,trim=10 0 10 8,clip]{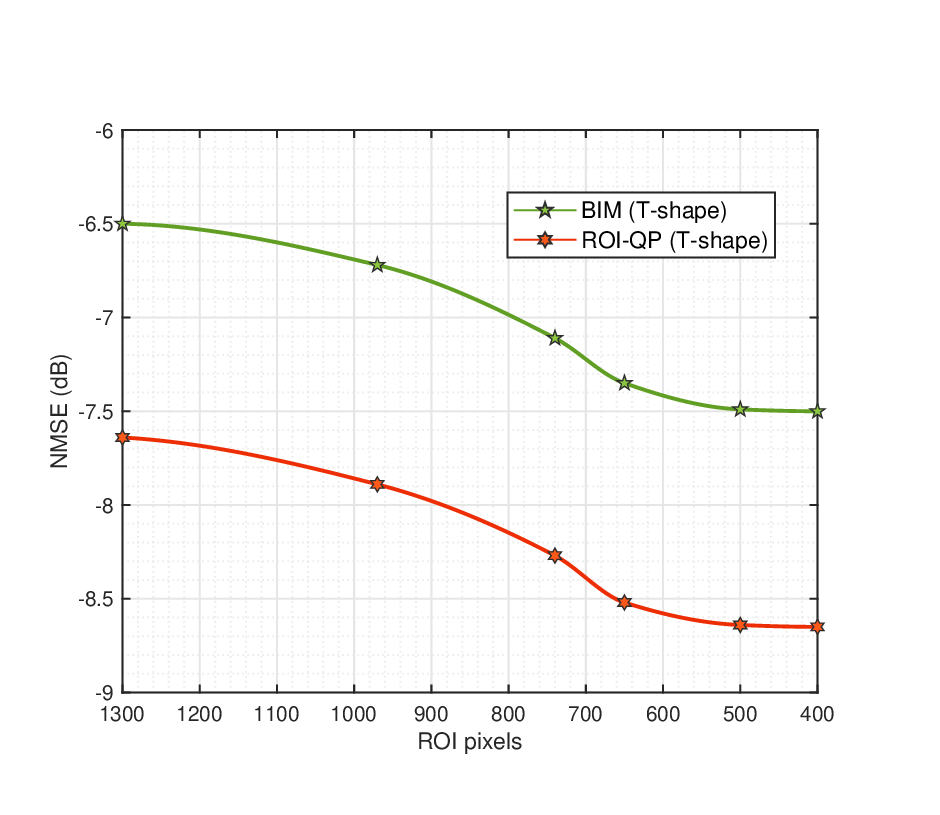}%
  }\vspace{-3mm}\hfill
  \subfloat[]{%
    \includegraphics[width=0.75\columnwidth,trim=10 0 10 8,clip]{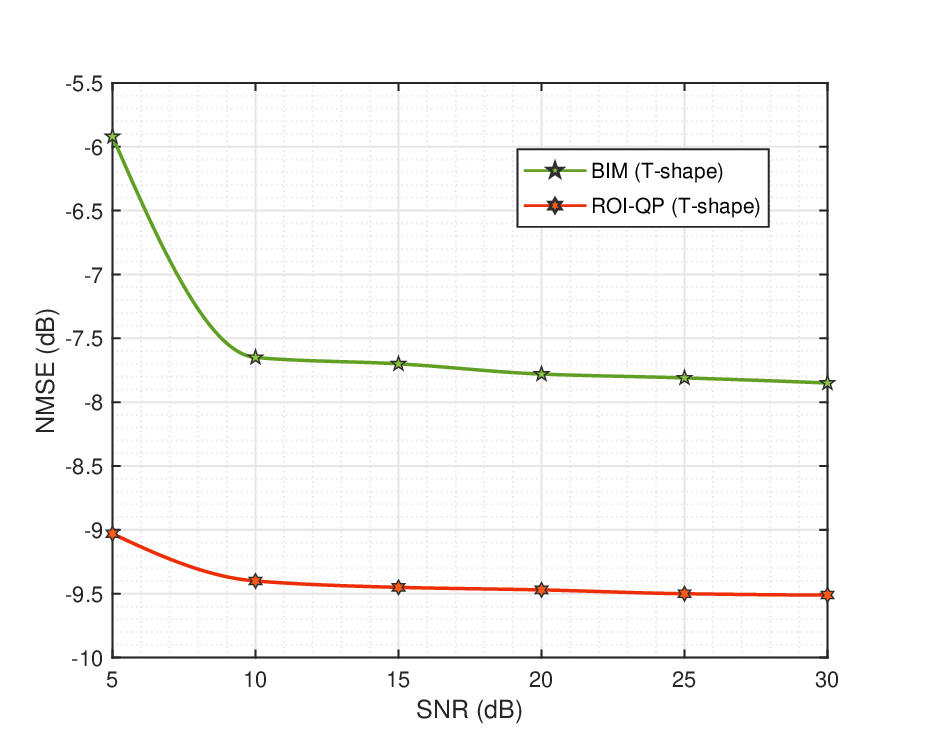}%
  }
  \caption{NMSE results for the T-shaped scatterer: (a) NMSE versus ROI pixels, and (b) NMSE versus SNR.}
  \label{fig:nmse-t}
  \vspace{2mm}
\end{figure}

%\vspace{-1em}
\appendices
\section{The Required Khatri-Rao Identity}

For any vectors $\mathbf x,\mathbf z\in\mathbb C^{m}$ and $\mathbf y,\mathbf w\in\mathbb C^{n}$, the following identities hold:

\begin{equation}\tag{A.1}
    (\mathbf{x} \circ \mathbf{y})^\mathrm{H}(\mathbf{z} \circ \mathbf{w}) = (\mathbf{x}^\mathrm{H} \mathbf{z})(\mathbf{y}^\mathrm{H} \mathbf{w}), 
\end{equation}

\begin{equation}\tag{A.2}
\| \mathbf{x} \circ \mathbf{y} \|_2 = \| \mathbf{x} \|_2 \| \mathbf{y} \|_2 .
\end{equation}

% =========================
% Appendix B
% =========================
\section{Air-Region Column Coherence--Continued}
For notational simplicity, we fix a tone $k$ and omit the subscript $k$ in this appendix. Under an (approximately) \ac{WSS} propagation model, channel responses observed at two
nearby spatial locations are \emph{locally coherent}: their second-order statistics depend
primarily on the displacement $\Delta r=r_j-r_i$, and the normalized spatial correlation
remains close to one when $\|\Delta r\|$ lies within the coherence distance. As a result,
the sub-vectors $\mathbf u_j$ and $\mathbf v_j$ associated with air-region pixels are
nearly colinear with $\mathbf u_i$ and $\mathbf v_i$, respectively. We formalize this
near-colinearity by the following scaling-plus-residual relations:
\begin{equation}\tag{B.1}
\begin{alignedat}{2}
\mathbf{u}_j &= \gamma_{ij}\,\mathbf{u}_i + \Delta {\mathbf{u}}, &\qquad
\|\Delta {\mathbf{u}}\|_2 &\le \varepsilon_u \|\mathbf{u}_i\|_2, \\
\mathbf{v}_j &= \rho_{ij}\,\mathbf{v}_i + \Delta {\mathbf{v}}, &\qquad
\|\Delta {\mathbf{v}}\|_2 &\le \varepsilon_v \|\mathbf{v}_i\|_2 .
\end{alignedat}
\end{equation}
Here, $\gamma_{ij},\rho_{ij}\in\mathbb C$ denote the best complex scalings (projections),
while $\varepsilon_u,\varepsilon_v\ll 1$ quantify the local de-coherence induced by the
spatial displacement.

Define
\begin{equation}\tag{B.2}
\begin{aligned}
U &\triangleq \|\mathbf{u}_i\|_2, \\
V &\triangleq \|\mathbf{v}_i\|_2, \\
\delta_{ij} &\triangleq |\gamma_{ij}|\,\varepsilon_v + |\rho_{ij}|\,\varepsilon_u + \varepsilon_u \varepsilon_v .
\end{aligned}
\end{equation}

We compute the $(i,j)$-th entry of the Gram matrix of $\mathbf A$.
By the Khatri--Rao property (A.1),
\begin{equation}\tag{B.3}
\langle \mathbf{a}_i, \mathbf{a}_j \rangle
= (\mathbf{u}_i^\mathrm{H} \mathbf{u}_j)\,(\mathbf{v}_i^\mathrm{H} \mathbf{v}_j).
\end{equation}
Substituting (B.1) into (B.3) gives
\begin{equation}\tag{B.4}
\begin{aligned}
\mathbf u_i^\mathrm{H}\mathbf u_j
&= \gamma_{ij}\|\mathbf u_i\|_2^2 + \mathbf u_i^\mathrm{H}\Delta\mathbf u, \\
\mathbf v_i^\mathrm{H}\mathbf v_j
&= \rho_{ij}\|\mathbf v_i\|_2^2 + \mathbf v_i^\mathrm{H}\Delta\mathbf v .
\end{aligned}
\end{equation}
By the Cauchy--Schwarz inequality and (B.1),
\begin{equation}\tag{B.5}
\begin{aligned}
\big|\mathbf u_i^\mathrm{H}\Delta\mathbf u\big|
&\le \|\mathbf u_i\|_2\|\Delta\mathbf u\|_2
\le \varepsilon_u U^2, \\
\big|\mathbf v_i^\mathrm{H}\Delta\mathbf v\big|
&\le \|\mathbf v_i\|_2\|\Delta\mathbf v\|_2
\le \varepsilon_v V^2.
\end{aligned}
\end{equation}
Combining the expansion of (B.3) with triangle inequality and the bounds in (B.5), we obtain
\begin{equation}\tag{B.6}
|\langle \mathbf{a}_i, \mathbf{a}_j \rangle|
\;\ge\;
\Big(|\rho_{ij}\gamma_{ij}|-\delta_{ij}\Big)\,U^2 V^2,
\end{equation}
where $\delta_{ij}$ is defined in (B.2).

From property (A.2),
\begin{equation}\tag{B.7}
\|\mathbf a_{i}\|_2 = \|\mathbf u_i\|_2\|\mathbf v_i\|_2 = UV,
\qquad
\|\mathbf a_{j}\|_2 = \|\mathbf u_j\|_2\|\mathbf v_j\|_2 .
\end{equation}
Moreover, by (B.1) and triangle inequality,
\begin{equation}\tag{B.8}
\|\mathbf u_j\|_2 \le (|\gamma_{ij}|+\varepsilon_u)U,
\qquad
\|\mathbf v_j\|_2 \le (|\rho_{ij}|+\varepsilon_v)V,
\end{equation}
hence
\begin{equation}\tag{B.9}
\|\mathbf a_{j}\|_2
\le (|\gamma_{ij}|+\varepsilon_u)(|\rho_{ij}|+\varepsilon_v)\,UV.
\end{equation}
The \ac{NCC} between the two columns is
\begin{equation}\tag{B.10}
\mu_{ij} \triangleq \frac{|\langle \mathbf{a}_i, \mathbf{a}_j \rangle|}
{\|\mathbf{a}_i\|_2 \,\|\mathbf{a}_j\|_2},
\end{equation}
and combining (B.6)--(B.10) yields
\begin{equation}\tag{B.11}
\mu_{ij}
\;\ge\;
\frac{|\rho_{ij}\gamma_{ij}|-\delta_{ij}}
{(|\gamma_{ij}|+\varepsilon_u)(|\rho_{ij}|+\varepsilon_v)}.
\end{equation}

For air-region pixels within a reasonably sized sensing region $D$, the large-scale fading
changes slowly over small displacements, while the small-scale fading remains locally
coherent due to limited angular spread. Consequently, when $\|\Delta r\|$ is within the
coherence distance, the residuals in (B.1) are small:
\begin{equation}\tag{B.12}
\varepsilon_u(\Delta r)\ll 1,\qquad \varepsilon_v(\Delta r)\ll 1 .
\end{equation}
Moreover, $|\gamma_{ij}|$ and $|\rho_{ij}|$ stay uniformly bounded away from zero in this
locally coherent regime, so the denominator of (B.11) is a bounded $\mathcal O(1)$ factor.
Therefore, (B.11) can be simplified to
\begin{equation}\tag{B.13}
\mu_{ij} \;\ge\; 1 - C\,\delta_{ij},
\end{equation}
where $C=\mathcal O(1)$ is a constant. Hence, $\mu_{ij}$ is close to $1$ for any two
air-region columns, implying strong inter-column coherence and, consequently, severe
ill-conditioning of the full-domain operator.

\ifCLASSOPTIONcaptionsoff
  \newpage
\fi


\begin{comment}
\begin{IEEEbiography}{Michael Shell}
Biography text here.
\end{IEEEbiography}

\begin{IEEEbiographynophoto}{John Doe}
Biography text here.
\end{IEEEbiographynophoto}

%\newpage

\begin{IEEEbiographynophoto}{Jane Doe}
Biography text here.
\end{IEEEbiographynophoto}
\end{comment}


%\vfill


\end{document}